\documentclass[hyperfootnotes=false]{iopjournal}
\usepackage[utf8]{inputenc}
\usepackage[english]{babel}
\usepackage{lmodern}
\usepackage[T1]{fontenc}
\usepackage{xcolor}
\usepackage{amsmath}
\usepackage{amssymb}
\usepackage{amsthm}
\usepackage{mathtools}
\usepackage{mathrsfs}
\usepackage{soul}
\usepackage{physics}
\usepackage{braket}
\usepackage{mdframed}

\usepackage{setspace}
\usepackage{bbm}
\usepackage{tcolorbox}
\usepackage{graphicx}
\usepackage{appendix}
\usepackage{relsize}
\usepackage{natbib}
\usepackage{algorithmic}
\usepackage{adjustbox}
\usepackage{booktabs}
\usepackage{ulem}
\usepackage{soul}
\usepackage{subcaption}

\newtheorem{definition2}{Lemma}

\newtheorem{Co}{Corollary}

\usepackage{tikz}
\usepackage{lipsum}
\usepackage{xcolor} 
\definecolor{myred}{RGB}{200,0,0}

\begin{document}
\title{Adaptive time Compressed QITE (ACQ) and its geometrical interpretation}

\author{
Alberto Acevedo$^{1,2}$,
Carmen G. Almudéver$^{3}$,
Miguel Angel Garcia-March$^{2}$,
Rafael Gómez-Lurbe$^{4}$,
Luca Ion$^{4}$,
Mohit Lal Bera$^{4,5}$,
Rodrigo M. Sanz$^{3}$,
Somayeh Mehrabankar$^{6}$,
Tanmoy Pandit$^{7,8,9}$,
Armando Pérez$^{4}$,
and Andreu Anglés-Castillo$^{3,*}$
}

\affil{$^1$Departamento de Matemática, Física y Ciencias Tecnológicas, Universidad CEU Cardenal Herrera, València, Spain}

\affil{$^2$Instituto Universitario de Investigación de Matemática Pura y Aplicada, Universitat Politècnica de València, València, Spain}

\affil{$^3$Departament d’Informàtica de Sistemes i Computadors, Universitat Politècnica de València, València, Spain}

\affil{$^4$Departament de Física Teòrica and IFIC, Universitat de València-CSIC, València, Spain}

\affil{$^5$ICFO–Institut de Ciències Fotòniques, The Barcelona Institute of Science and Technology, Av. Carl Friedrich Gauss 3, 08860 Castelldefels (Barcelona), Spain}

\affil{$^6$Queensland Quantum and Advanced Technologies Research Institute, Griffith University, Yuggera Country, Brisbane, QLD 4111, Australia}

\affil{$^7$VTT Technical Research Centre of Finland, Tietotie~3, Espoo, Finland}

\affil{$^8$Institute for Theoretical Physics, Leibniz Institute of Hannover, Hannover, Germany}

\affil{$^9$Institute for Physics and Astronomy, TU Berlin, Germany}

\affil{$^*$Author to whom any correspondence should be addressed.}

\email{aangcas@upv.es}

\begin{abstract}
    Imaginary Time Evolution (ITE) is a well-established method for ground-state preparation, a fundamental problem in many fields such as materials science, chemistry, and optimization. Quantum Imaginary Time Evolution (QITE) approximates this evolution on quantum hardware but suffers from high circuit depth and numerous measurements. In this work we introduce Adaptive-time Compressed QITE (ACQ), a novel algorithm that reduces resource-cost by combining adaptive time steps with circuit compression. This approach leverages geometric insights by characterizing its relationship to geodesic trajectories with a measure that distinguishes trajectories in $\mathbb{CP}^N$. Recalling that ITE is a gradient flow on the complex projective plane $\mathbb{CP}^N$, such trajectory measures allow one to measure the deviation from geodesicity of said flow. For Hamiltonians with only two distinct eigenvalues (spectral cardiality), ITE and QITE exactly trace geodesics, this fact motivates an adaptive strategy for systems whose corresponding spectral cardinality is greater than 2,  where QITE unitaries are reused until an energy increase signals departure from the ITE path. This is implemented via a line search for energy minimization. Circuit compression is achieved by approximating the sequence of QITE unitaries with a single element of a one-parameter group. Numerical simulations on the Transverse Field Ising Model and the Heisenberg model demonstrate that ACQ achieves comparable fidelity to standard QITE while significantly reducing the number of QITE optimizations and maintaining fixed circuit depth during propagation. Gate-count estimates and an analysis of the fidelity scaling  with truncation parameters are provided. A gate count and performace comparison with the state of the art method Double Bracket QITE is also performed.
\end{abstract}

%\tableofcontents
\section{Introduction}\label{sec:intro}

Ground state preparation is one  of the most important applications of quantum computers.  Many proposals aim at taking advantage of quantum mechanics  to find ground states in  various fields, e.g.  quantum simulators in material science~\cite{cirac2012goals,bernien2017probing}, problems in quantum chemistry~\cite{mcardle2020quantum}, and mapping optimization problems in Hamiltonians attainable for quantum computers~\cite{Lucas2014,Hadfield2021}.  While approximate ground state preparation with classical methods~\cite{NAKATANI2018,RevModPhys.93.045003,suzuki1976generalized} has been extensively explored, there exist proposals of quantum algorithms that may present an advantage.  These include quantum-classical variational algorithms such as quantum approximate optimization algorithm~\cite{farhi2014quantum,farhi2014quantumMaxCut} and the variational quantum eigensolver~\cite{peruzzo2014variational}, that are proposals with reduced quantum resources targeted for the NISQ era~\cite{preskill2018quantum}, or adiabatic state preparation~\cite{farhi2000}. 

The focal point of this paper is Imaginary Time Evolution (ITE) \cite{suzuki1976generalized} methods and how they can be implemented on quantum hardware. ITE is one of a variety of so-called dynamical optimization methods. Even though ITE produces globally non-linear dynamics, it preserves the norm and smoothness at every instance of time; hence it is locally characterized by linear unitary evolution. This is the basis of Quantum Imaginary Time Evolution (QITE)~\cite{mott} methods that approximate ITE with a family of unitary processes to arbitrary precision and the basis that motivates optimization strategies here presented. 
Many proposals that optimize the performance and resource cost of QITE are presented in Section~\ref{sec:QITEopts}. In this paper, we introduce in Section \ref{sec:adapQITE} an Adative-time Compressed QITE (ACQ) that finds the optimal time for energy reduction of the family of QITE unitaries by a line search. This optimization reduces the number of measurements required for convergence, while a Trotter compression of these unitaries is performed to reduce the resulting circuit depth.

This paper is structured as follows. In the rest of Section \ref{sec:intro} the preliminaries for ITE methods are presented, the original QITE proposal for its implementation in a quantum computer, and the alternatives in the literature that optimize it. 
In Sec.~\ref{sec:geom} the geometric properties of ITE and QITE are reviewed. To better understand the underlying structure of ITE, one may also compare its dynamics to geodesic evolution in  $\mathbb{CP}^N$ by defining a geometrically-motivated distinguishability measure, which shall also be used to compare the deviation of QITE (and ACQ) from ITE.  This geometric treatment serves as motivation for the novel algorithm presented in Sec.~\ref{sec:adapQITE}. In Sec.~\ref{sec:results} numerical results that support the effectiveness of this algorithm are presented with concluding remarks being left for Sec.~\ref{sec:conclusion}.

%\section{Review of Imaginary Time Evolution Methods}

%In this section we give the preliminaries for Imaginary Time evolution and review the Quantum Imaginary Time Evolution proposals that can be implemented on quantum hardware. 

\subsection{Imaginary Time Evolution}\label{subsec:ITE}
By performing a Wick-Rotation \cite{wick1954properties} on unitary evolution, generated by a Hamiltonian $\boldsymbol{\hat{H}}$, one gets the so-called imaginary time $\tau = i t$ evolution ($\tau$ is assumed to be real forcing $t$ to be imaginary; hence the name). One can evolve an initial state $|\psi_{0}\rangle$ as follows
\begin{equation}
\label{eqn:ite}
    \ket{\psi(\tau)}= e^{-\tau\boldsymbol{\hat{H}} } \ket{\psi_0}= \sum_n c_n e^{-\tau E_n } \ket{E_n}~,
\end{equation}
where the initial state has been decomposed in the eigenbasis of $\boldsymbol{\hat{H}}$ with $c_n:=\braket{E_n|\psi_0}$. One can see that, at long times, the exponential with the smallest $E_n$ dominates the above expansion. For that reason, if the initial state has some support with the ground state $\braket{E_0|\psi_0}\neq0$, at long times $\ket{\psi(\tau)}$ will be approximately on the same ray as $\ket{E_0}$, up to a normalization factor. To keep the evolved state pure and ensure numerical stability in computational scenarios, ITE is defined with the following normalization factor
\begin{equation}\label{eq:ITE}
    \ket{\psi(\tau)}= \frac{e^{-\tau\boldsymbol{\hat{H}}}\ket{\psi_0}}{\norm{e^{-\tau\boldsymbol{\hat{H}}}\ket{\psi_0}} }~,
\end{equation}
where $\norm{\cdot}$ denotes the 2-norm in Hilbert space.
Equation~\eqref{eq:ITE} is a solution to the Wick-Schrödinger equation
\begin{equation}\label{eq:sch-wick}
\partial_{\tau}|\psi(\tau)\rangle = -(\boldsymbol{\hat{H}}-E(\tau))|\psi(\tau)\rangle~,
\end{equation}
where $E(\tau) := \braket{\psi(\tau)|\boldsymbol{\hat{H}}|\psi(\tau)}$ is the expectation value of the energy. 
The evolution  in Eq.~\eqref{eq:ITE} can be split into $n$ discrete time steps of length $\Delta \tau$
\begin{equation}
\ket{\psi(n\Delta\tau)}=\frac{e^{-n\Delta\tau \boldsymbol{\hat{H}}}\ket{\psi_0}}{\norm{e^{-n\Delta\tau \boldsymbol{\hat{H}}}\ket{\psi_0}} }~,
\end{equation}
such that the state $\ket{\psi_n}\equiv\ket{\psi(n\Delta \tau)}$ at step $n$ is related to the forward step as
\begin{equation}\label{eq:ITE_step}
    \ket{\psi_{n+1}}=\frac{e^{-\Delta \tau\boldsymbol{\hat{H}} }\ket{\psi_n}}{\norm{e^{-\Delta \tau\boldsymbol{\hat{H}} }\ket{\psi_n}}}~.
\end{equation}
Here, a fixed time-step $\Delta\tau$ has been utilized . In what is to come, the problem of finding adaptive time-steps $\Delta\tau_{n}$ with the goal of reducing iteration count shall be explored. 

\subsection{Quantum Imaginary Time Evolution}\label{sec:QITE}

The map $M_{\tau}(\cdot):=\frac{e^{-\tau\boldsymbol{\hat{H}}}(\cdot)}{\norm{e^{-\tau\boldsymbol{\hat{H}}}(\cdot)} }$ mapping quantum states to quantum states is non-unitary and non-linear; however, quantum computers directly implement unitary evolution and projective measurements; 
non-unitary evolution must be simulated through sequences of unitary gates. The goal of QITE is to approximate such non-unitary dynamics with unitary operations apt to be executed in a quantum computer. The seminal QITE proposal by Motta et al.  \cite{Motta2019DeterminingEA} introduced a hybrid quantum-classical 
approach consisting on finding unitary operators $\boldsymbol{\hat{U}}_n=e^{-i \boldsymbol{\hat{A}}_n \Delta\tau}$ that closely approximate the time step evolution generated by the operator of Eq.~\eqref{eq:ITE_step}. To this aim, the unitary generator $\boldsymbol{\hat{A}}_n$ can be expressed in terms of an operator basis
\begin{equation}\label{eq:generatorQITE}
\boldsymbol{\hat{A}}_n=\sum_I a_I(n) \boldsymbol{
    \hat{\sigma}}_I~,
\end{equation}
where $\boldsymbol{\hat{\sigma}}_I$ is the operator basis, which, for instance, for spin systems would be composed of Pauli strings.\footnote{Unless otherwise specified, this work will be reviewing spin chain systems with finite correlation lengths, but the same arguments apply to other kinds of many particle systems.} In \cite{Motta2019DeterminingEA} the problem is translated into an optimization where the $a_I(n)$ coefficients are found by minimizing the state norm difference between the state evolved as in Eq.~\eqref{eq:ITE_step} and evolved with the unitary $\boldsymbol{\hat{U}}_n$; the precision of these algorithms of course depend on the time-steps $\Delta\tau$ being small enough. See \cite{understandingQITE} for a pedagogical and detailed review of the procedure followed by QITE to generate the unitaries $\boldsymbol{\hat{U}}_n$.

%revisar la explicació del liniar system i la exponencial
While this procedure may yield dynamics that are arbitrarily close to ITE, the size of the linear system of equations to solve increases exponentially with the size of the system. 
To tackle this problem, the Hamiltonian is divided into $K$-local terms 
\begin{equation}\label{eq:HamPieces}
    \boldsymbol{\hat{H}} = \sum_{k=1}^{N_K}\boldsymbol{\hat{h}}_k~,
\end{equation}
where each $\boldsymbol{\hat{h}}_k$ acts on at most $K$ neighbouring particles (qubits), so that the imaginary evolution may be Trotterized as follows
\begin{equation}
    e^{-\tau\boldsymbol{\hat{H}} } = \left( \prod_k e^{-\Delta \tau\boldsymbol{\hat{h}}_k } \right)^n+O(n \Delta\tau^2)~,
\end{equation}
where the evolution is divided into discrete time steps as before $\tau=n\Delta \tau$. Next, unitaries $\boldsymbol{\hat{U}}_{n,k}:=e^{-i\Delta \tau \boldsymbol{\hat{A}}_{n,k}}$ approximating the ITE generated by the individual Hamiltonian pieces,
\begin{equation}\label{eq:QITE_cond}
    e^{-i\Delta\tau \boldsymbol{\hat{A}}_{n,k}}\ket{\psi_n}\sim \frac{e^{-\Delta\tau\boldsymbol{\hat{h}}_k}\ket{\psi_n}}{\norm{e^{-\Delta\tau\boldsymbol{\hat{h}}_k}\ket{\psi_n}}}~,
\end{equation}
need to be found.

If $K$ is fixed at some small value, then the exponential on the right can be easily computed; however, owing to the fact that this operation introduces correlations beyond the $K$ neighbouring particles, the unitary operation would need to be computed in a larger domain in order to effectively capture all of the correlations\footnote{Recall that ITE is a non-unitary evolution, and even though $\boldsymbol{\hat h}_k$ acts locally on $K$ qubits, the r.h.s. of Eq.~\eqref{eq:QITE_cond} does not necessarily only act on it.}.
At most, the unitary would act on the whole domain of $\boldsymbol{\hat{H}}$ which, for large systems, is an unfeasible task, whence a truncation of the domain that $\boldsymbol{\hat{A}}_{n,k}$ acts on is introduced. 
Let the domain of $\boldsymbol{\hat{U}}_{n,k}$, where it acts not trivially, be labelled as $D$; $D \geq K$ represents the number of qubits on which the unitary acts non trivially.
%\textbf{While $\boldsymbol{\hat{h}}_k$ has locality $K$, the resulting unitary requires domain $D \geq T$ to capture correlations introduced by the imaginary time evolution (repeating yourself)}. 
The computational cost scales exponentially with $D$.
Ideally, $D$ should be truncated to become the value of the correlation distance of the system of interest, see Fig.~1 in \cite{Motta2019DeterminingEA} for further details; a deeper discussion regarding the parameter $D$ is left for section \ref{geosec} where the underlying geometric notions vital to this parameter will be further motivated and developed. 
Additionally, the Trotterization allows one to synthetize the $K$-local unitaries $\boldsymbol{\hat U}_{n,k}$ into a much shallower circuit than the one that would be obtained for the previous non-local unitary $\boldsymbol{\hat U}_n$.

\subsection{QITE optimizations}\label{sec:QITEopts}
There are various proposals for tackling QITE's shortcomings in practical scenarios. For instance, reducing the circuit-depth of generated circuits with a reverse Suzuki-Trotter decomposition, (i.e. compression) can be performed on all the unitaries generated by QITE. The methods that make use of this optimization are referred to as compressed QITE (cQITE) \cite{Nishi2021} or step-merged QITE (smQITE) \cite{Gomes2020}. Each step of QITE corresponds to a parameter update of the circuit, however, the measurement of expectation values (runtime cost) at each time step and the solution of the linear system that minimizes the difference of terms in Eq.~\eqref{eq:QITE_cond} (classical cost) still needs to be performed at each step. 

An equivalent formulation of Eq.~\eqref{eq:sch-wick} is given by
\begin{equation}\label{eq:DB-ITE-ODE}
    \partial_\tau \ket{\psi(\tau)} = [\boldsymbol{\hat{\rho}}(\tau),\boldsymbol{\hat{H}}] \ket{\psi(\tau)}
\end{equation}
where $\boldsymbol{\hat{\rho}}(\tau)=\ket{\psi(\tau)}\bra{\psi(\tau)}$, which for small time steps can be rewriten as the recursion
\begin{equation}\label{eq:DB-ITE}
    \ket{\psi_{n+1}} \approx e^{\Delta \tau [\boldsymbol{\hat{\rho}}_n,\boldsymbol{\hat{H}}]}\ket{\psi_n}~,
\end{equation}
where $\boldsymbol{\hat \rho}_n=\ket{\psi_n}\bra{\psi_n}$, which is by itself a unitary evolution\footnote{Notice that the commutator of two Hermitian operator is anti-Hermitian, so that the exponential is unitary. Also note that the unitary generator depends non-linearly with the state it is evolving.} that approximates ITE. This fact is exploited in \cite{glu} to synthesize a circuit that implements this version of QITE, which is called Double Bracket QITE (DB-QITE). The benefits of this method are that it does not rely on intermediate measurements \cite{Iqbal2024topological} to compute the unitary of the following step, it is a pure quantum algorithm. This is achieved by nesting the unitaries used to evolve one step into the unitaries used to evolve to the next one, 
which makes the gate count to increase exponentially with the number of steps in the algorithm. We explore the resource cost of this novel algorithm in the following and compare it with the proposal presented in this paper.

A variational formulation of imaginary-time evolution based on McLachlan’s variational principle was proposed as an alternative approach to approximating the imaginary-time-evolved quantum state in Eq.~\eqref{eq:ITE}~\cite{mcardle2019variational}. Rather than approximating each time step by an explicit unitary, as in QITE, Variational Quantum Imaginary Time Evolution (VarQITE) projects the imaginary-time dynamics onto a parameterized manifold of quantum states defined by a variational ansatz $\ket{\phi(\boldsymbol{\theta})}$. In this framework, the evolution of the quantum state is mapped onto the evolution of the variational parameters.

This variational formulation reveals a direct geometric connection between VarQITE and the Quantum Natural Gradient (QNG) method~\cite{Stokes2020quantumnatural}, whereby VarQITE can be interpreted as a natural-gradient descent on the energy landscape of the variational manifold. The QNG was originally introduced in the context of variational quantum algorithms to improve the optimization of parametrized quantum circuits by incorporating geometric information about the underlying Hilbert space. Numerical studies have demonstrated that QNG-based methods can accelerate convergence~\cite{Stokes2020quantumnatural}, help avoid certain local minima~\cite{wierichs2020avoiding}, and exhibit strong robustness to random initializations~\cite{62wx-tvk5}. Moreover, these advantages have been shown to persist even in the presence of noise~\cite{62wx-tvk5,koczor2022quantum}.

Despite their promise, both QNG and VarQITE are often limited by the substantial quantum and classical resources required to compute the Quantum Geometric Tensor. To mitigate this computational overhead, several resource-efficient estimation strategies have been proposed~\cite{gomezlurbe2025efficientprotocolestimatequantum,PhysRevA.111.012424}. In addition, VarQITE can suffer from the barren plateau phenomenon~\cite{McClean_2018,Larocca_2025}, in which the cost-function landscape becomes exponentially flat as the system size or circuit depth increases. In such regimes, gradients with respect to the variational parameters vanish on average, leading to extremely small parameter updates and causing the optimization to stall.

Truncating the QITE unitaries by restricting them to a domain $D$ is a generic truncation procedure. There are some problem specific optimizations that can be exploited to obtain truncations tailored to specific problems. For instance, in combinatorial problems \cite{Alam2023,PhysRevA.109.052430} a separable ansatz for Eq.~\eqref{eq:generatorQITE} is considered, since the problem is mapped to a Hamiltonian with a separable ground state, which allows to greatly reduce the problem size. In \cite{MolecularImprov}, where the ground states of molecular systems is pursued, the construction of $\boldsymbol{\hat{A}}_n$ is performed by considering only anti-Hermitian fermionic operators and truncating the higher order ones, which are known to have small contributions in these systems. For combinatorial problems, a hybrid VQE-cQITE has been proposed in \cite{xie2025adaptiveweightedqitevqealgorithm} to overcome the shortcoming of each procedure by a weighted combination of the update rule of each method.

In this work an algorithm that simultaneously reduces the circuit depth of the circuits used to approximate the ITE protocol and decreases the number of calls to the QITE subroutine during the optimization process is proposed. With such an algorithm one may also lower the classical computational overhead by reducing the total number of expectation values that need to be evaluated for energy minimization with respect to the one-parameter unitary group generated in each QITE step. In the following, a review of the geometric properties of ITE methods that motivate it is presented.

% \paragraph{Paper structure} In Section \ref{sec:geom} we review the geometric aspects of ITE and QITE, while in section \ref{sec:adapQITE} we exploit these geometric properties to propose an adaptive QITE that aims to reduce some of the pain points of QITE. In Section \ref{sec:results} we provide  numerical results that support the proposed approach and we finish with some conclusive remarks and possible future work in Section \ref{sec:conclusion}.

\section{Geometric Background and Motivation}\label{sec:geom}

Although ITE, described in Eq.~\eqref{eq:ITE}, produces non-unitary dynamics, the normalization in the denominator keeps states in the space of pure states, i.e. $\mathbb{CP}^N$. Furthermore, it has been shown in multiple references \cite{SciPostPhys.9.4.048,na} that ITE is a Riemannian gradient descent flow pertaining to the minimization of the functional $E(\psi)=\braket{\psi|\boldsymbol{\hat{H}}|\psi}$ over the complex projective plane. These trajectories are steepest descent trajectories and converge to the ground state, which is the critical point in $\mathbb{CP}^{N}$ corresponding to the global and only minimum of the functional $E(\psi)$. The rest of the critical points are the excited energy subspaces of the Hamiltonian and correspond to relative maxima \cite{SciPostPhys.9.4.048,na}. 

The distance between states in $\mathbb{CP}^N$ is the intrinsic metric known as the Fubini-Study distance
\begin{equation}\label{eq:fubini_dist}
	d_{FS}(|\psi\rangle, |\phi\rangle) := \arccos(|\langle \phi|\psi\rangle|)~.
\end{equation}
Note that, by definition, the distance between any two points is equal to the length of the geodesic that connects them, which is, indeed, the minimum achievable by any path between them.
Geodesics on  $\mathbb{CP}^N$ between any two points minimize this distance. 
In general these paths are given by 1-parameter unitary groups, but different parametrizations of the geodesics exist. 

It is worth noting that the path followed by ITE for the case where the Hamiltonian in question has spectral cardinality 2 is that of a geodesic \cite{arxivgluzagrover} between the initial state and the ground state; for higher spectral cardinality Hamiltonians this is no longer the case. These notions are developed further in Appendix~\ref{sec:app_geometric}. Furthermore, in Appendix~\ref{sec:1qbit_geo_QITE} it is shown that QITE also follows a geodesic trajectory for the case of 1 qubit. One might intuit at this point that the departure of the (Q)ITE trajectory from a geodesic path might be greater for higher dimensional systems. To assess this, in what follows,  a measure that precisely quantifies how far apart are two different trajectories in $\mathbb{CP}^N$ is introduced.

\subsection{Distingushing trajectories on \texorpdfstring{$\mathbb{CP}^{N}$}{CP N}}
\label{sec:geom1}

%ITE generates trajectories on $\mathbb{CP}^{N}$ which are in general not geodesics. Nevertheless, as already mentioned, it has been shown in \cite{arxivgluzagrover} that for the special case of rank-2 Hamiltonians the trajectories followed by ITE and the geodesic connecting the initial state to the  ground state coincide. This being the case, it would be interesting to define a measure between trajectories in order to analyze deviation from agreement and to see how it relates to dimension and spectral gap
%future work???
%of the Hamiltonian amongst other things.

Let $\Gamma(\mathbb{CP}^{N})$ be the set of all smooth trajectories on $\mathbb{CP}^{N}$. Next, let us equip this set with the following distinguishability measure. 
Let $S_1$ and $S_2$ be two sets constituting paths in $\mathbb{CP}^{N}$ and let $\gamma_{1}(\tau)$, $\gamma_{2}(\tau)$, $\tau \in [0,T]$ be parametrizations of the paths $S_{1}$ and $S_{2}$ respectively . 
We then define,
\begin{align}
\label{eq:traj_distance} 
D_f(S_1,S_2) & =\int_{0}^{T}\arccos\big(|\langle\gamma_{1}(\tau),\gamma_{2}(\tau)\rangle|\big)d\tau ~.
\end{align}
This integral measures the totality of differences, in the sense of the intrinsic  Fubini-Study metric, between points in trajectory $S_{1}$ and those of $S_{2}$. Indeed, this integral will be zero if and only if the trajectories $S_1$ and $S_2$ follow the same path and move at the same speed.

\subsection{Velocity Measure}
\label{sec:geom2}
Since the trajectories $\gamma_{1}(\tau)$ and $\gamma_{2}(\tau)$ are not only different in general but also move at different velocities, a more complete analysis would involve studying how similar the velocities of these two trajectories are; velocity here is meant to be interpreted here as the generator of local-unitary dynamics at any given time $\tau$, i.e. the effective Hamiltonian of the dynamics at said time. Whence the following velocity measure follows.

\begin{equation}
     \mathcal{D}_{v}(S_1,S_2)  := \Big(\int_{0}^{T}
\big\lVert \dot\gamma_{1}(\tau)-\dot\gamma_{2}(\tau)\big\rVert_{\mathrm{FS}}^{2}\,d\tau\Big)^{1/2}~,
\end{equation}
where the Fubini-Study speed is defined as follows
%I propose changing the naming convention to v_{FS}(\gamma) = sqrt{...}, then in the intetral have v_{FS}(\gamma_1-\gamma_2)
\begin{equation}
\big\lVert \dot\gamma(\tau)\big\rVert_{\mathrm{FS}}
:=\sqrt{\,
\langle \dot{\gamma}(\tau)| \dot{\gamma}(\tau)\rangle
\;-\;
\big|\langle \gamma(\tau)| \dot{\gamma}(\tau)\rangle\big|^{2}
\,}~,
\end{equation}
where $|\gamma(\tau)\rangle$ is a lift
%what is meant by a lift?
of the trajectory $\gamma(\tau)$ at time $t$, e.g.
if the evolution is unitary, whence
\begin{align}\label{eq: FS velocity}
\partial_{\tau}|\gamma(\tau)\rangle &= -i\,H\,|\gamma(\tau)\rangle~, \\
\big\lVert \partial_{\tau}|\gamma(\tau)\rangle\big\rVert_{\mathrm{FS}}
&= \sqrt{\langle H^{2}\rangle_\tau - \langle H\rangle_\tau^{2}}
= \Delta H_{\tau}~.
\end{align}
as expected.
For the case where trajectories $\gamma_{1}(\tau)$ and $\gamma_{2}(\tau)$ are equal at $\tau=0$, it can be shown that $\mathcal{D}_{v}(S_{1},S_{2})=0$ if and only if the curves are the same. In general it will be more challenging to calculate the velocity distance $\mathcal{D}_{v}$. However, this distance has a lower bound that is very simple to estimate. Namely the following, 
\begin{align}
\label{eqn:speed}
\mathcal{D}_{v}(S_1,S_2) &:= \bigg(\int_{0}^{T}
\big\lVert  \dot\gamma_{1}(\tau)-\dot\gamma_{2}(\tau)\big\rVert_{\mathrm{FS}}^{2}\,d\tau\bigg)^{1/2}\geq
\frac{|L_{1}-L_{2}|}{\sqrt{T}}~,\\
\label{eq:length} 
L(S) & = \int_0^T  \big\lVert \dot\gamma(\tau)\big\rVert_{\mathrm{FS}} d\tau~,
\end{align}
where $L_{1}$ and $L_{2}$ are the lengths (on $\mathbb{CP}^{N}$) of the trajectories $\gamma_{1}(\tau)$ and $\gamma_{2}(\tau)$ respectively, and can be expressed in terms of Eq.(\ref{eq:length}). One may therefore utilize this bound to see how the divergence between trajectories $S_{1}$ and $S_{2}$ (e.g., ITE and ACQ) depends on the dimensionality of the system, in a purely geometric sense. It is worth noting here that for gradient flows and respective approximate line search methods (ITE and ACQ for example), the respective trajectory lengths coincide if and only if the line search is exactly the gradient flow. Hence, the measure $|L_{1}-L_{2}|$ is a relevant measure of difference in its own right even outside of the context of the measures $\mathcal{D}_{f}(S_{1},S_{2})$ and $\mathcal{D}_{v}(S_{1}, S_{2})$. 

To showcase the behaviour of the measure defined in Eq.~\eqref{eqn:speed} let us first consider the Hamiltonian of the Transverse Field Ising Model (TFIM)
%mention if periodic o open conditions
\begin{equation}\label{eq:TFIM}
    \boldsymbol{\hat{H}} = \sum_{j=1}^N J\boldsymbol{\hat{\sigma}}^z_j \boldsymbol{\hat{\sigma}}^z_{j+1} +\sum_{j=1}^N h \boldsymbol{\hat{\sigma}}^x_j~,
\end{equation}
where $J$ is the nearest neighbour interaction strength and $h$ the external field strength.
In the left panel of Fig.~\ref{fig:dist_ITE_GEO} one may appreciate how this distance measure deviates when comparing the ITE trajectory with a geodesic one (connecting the same initial state and the ground state of $H$). It can be seen in green that the distance increases with system size $N$, indicating that the ITE trajectory deviates more from the geodesic path for larger systems. This distance would be exactly zero for 1 qubit systems, but the Hamiltonian $\eqref{eq:TFIM}$ is ill-defined for $N=1$. The path length of ITE and the geodesic is also shown for illustrative purposes. While the geodesic length asymptotically reaches the maximum distance in $\mathbb{CP}^{2^N}$ which is $\pi/2$, the ITE trajectory length seems to keep increasing with $N$. A naive explanation for this behaviour can be inferred from the bounds on the ITE convergence time $\tau_{\text{conv}}$ derived in Appendix \ref{sec:app_Length}. In particular, the lower-bound scaling of the expectation value, $\mathbb{E}(\tau_{\mathrm{conv}}) \in \Omega\left( \frac{N}{E_n - E_0} \right)$, suggests that increasing the number of qubits generally slows down the convergence, which can lead to longer trajectory lengths.

\begin{figure}
    \centering
    \includegraphics[width=\linewidth]{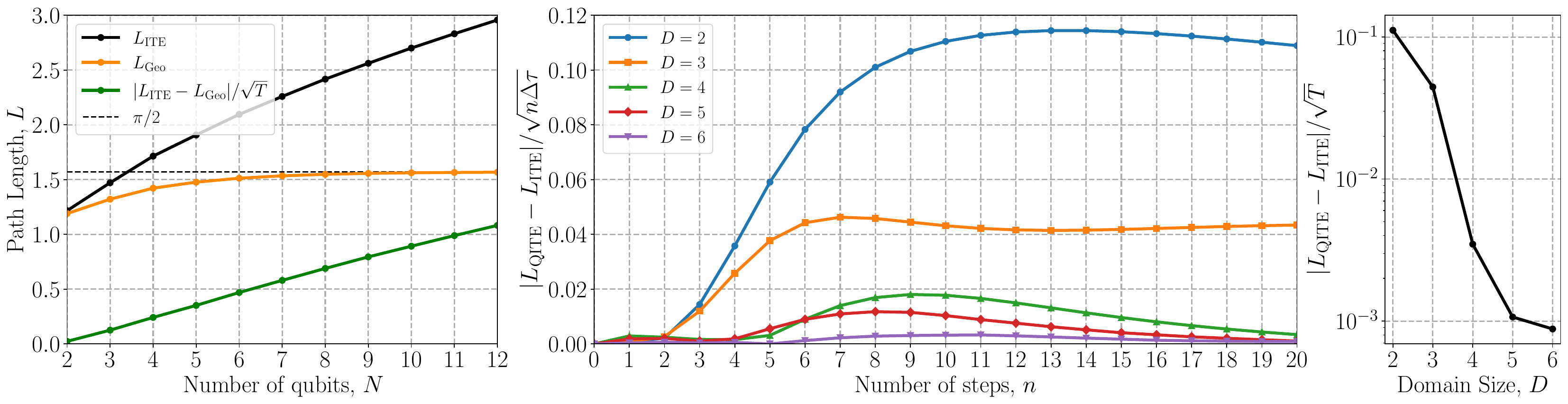}
    \caption{\textbf{Left:} Upper bound of the trajectory distance $\mathcal{D}_v(S_\mathrm{Geo}, S_\mathrm{ITE})$ defined in Eq.~\eqref{eqn:speed} in green between the ITE trajectory and the geodesic connecting the initial state $\ket{\psi_0}=\ket{0}^{\otimes N}$ and the ground state $\ket{E_0}$ of the Transverse Field Ising Model (TFIM) with parameters $J = 0.5$ and $h = 1$. The path length of the geodesic and ITE path are respectively plotted in orange and black; note that the geodesic path length approaches asymptotically the value $\pi/2$ for high $N$.
    \textbf{Centre:} Path distance between QITE and ITE at intermediate points of the evolution for the TFIM with the same parameters and $N=8$ qubits. After 20 steps the state generated by (Q)ITE is already very close to the ground state. 
    \textbf{Right:} Path distance between QITE and ITE for $N=8$ qubits and different values of the domain size $D$.}
    \label{fig:dist_ITE_GEO}
\end{figure}

\subsection{Dependence on the Locality Parameter \texorpdfstring{$D$}{D}}
\label{geosec}
When the locality of the operators $\boldsymbol{\hat{A}}_{n,k}$ defined in Eq.~\eqref{eq:QITE_cond} is restricted, one shifts from the problem of estimating a gradient descent flow using unitary operators generated by time–independent generators $\boldsymbol{\hat{A}}_{n,k}$ with arbitrary support, to one where their support is assumed to be of finite length $D$; i.e.\ these generators may only act on up to $D$ qubits at a time.

From the geometric perspective developed previously, unrestricted QITE corresponds to approximately following the Riemannian gradient flow of the energy functional
\begin{equation}
    E(\psi)=\langle\psi|\boldsymbol{\hat{H}}|\psi\rangle~,
\end{equation}
on the manifold $\mathbb{CP}^N$. Restricting the domain size to $D$ therefore has the following interpretation. Instead of evolving along the full gradient vector field on $\mathbb{CP}^N$, QITE follows the \emph{projection} of the gradient onto the submanifold of states reachable from $\ket{\psi_n}$ by $D$–local unitaries. In this sense, the parameter $D$ determines the dimension of the accessible tangent subspace at each step, and hence acts as a geometric resolution parameter for how well the exact ITE direction is approximated.

The small–time–step imaginary–time evolution that the generator $\boldsymbol{\hat{A}}_{n,k}$ approximates is in general not local. Indeed, the right–hand side of Eq.~\eqref{eq:QITE_cond} has an effective generator which is non–linear in both time and $\ket{\psi_n}$ and, for small $\Delta\tau$, is approximately of the form
\[
[\boldsymbol{\hat{h}}_k, \ket{\psi_n}\!\bra{\psi_n}]~.
\]
For this commutator to preserve the locality of $\boldsymbol{\hat{h}}_k$, the state $\ket{\psi_n}\!\bra{\psi_n}$ itself would need to be local, which is not generally true. As imaginary time evolves, correlations spread through the system and the projector $\ket{\psi_n}\!\bra{\psi_n}$ becomes increasingly non–local, even when $\boldsymbol{\hat{h}}_k$ is strictly local. The non-locality of $\ket{\psi_n}\!\bra{\psi_n}$ is quantified in terms of correlation length; in particular one says that a state $\ket{\psi}$ has correlation length $\xi$ if its associated correlation functions satisfy
$\big|\langle \boldsymbol{\hat{O}}_X \boldsymbol{\hat{O}}_Y \rangle - \langle \boldsymbol{\hat{O}}_X \rangle\langle \boldsymbol{\hat{O}}_Y \rangle\big|
\le C e^{-\mathrm{dist}(X,Y)/\xi}$ for any local observables $\boldsymbol{\hat{O}}_X,\boldsymbol{\hat{O}}_Y$. Here, it is assumed that the observables are supported over the subset of some lattice spins, namely $X$  and $Y$, labelling which subset of spins they act on non-trivially. For a spin-chain the lattice is of course just $\mathbb{Z}_{N}$ and $X$, $Y \subset \mathbb{Z}_{N}$. This is indeed a strong assumption and such bounds on the correlations functions are not in general guaranteed, this problem is partially addressed in \cite{Matt}.

Under the finite–correlation–length assumption of \cite{Motta2019DeterminingEA}, the operator $[\boldsymbol{\hat{h}}_k, \ket{\psi_n}\!\bra{\psi_n}]$ is quasi–local, in the sense that its expansion in Pauli strings contains coefficients that decay exponentially with the size of the support, i.e. string-length. In that regime, increasing $D$ improves the approximation systematically by capturing a larger fraction of the tangent space where the gradient lives. In particular, one expects good performance whenever $D$ exceeds the instantaneous correlation length of $\ket{\psi_n}$. 

With these considerations in mind, it is worthwhile to study how efficient ACQ and QITE respectively are when the locality parameter $D$ is varied. Since deducing generic analytic relationships between $\Delta\tau$ and $D$ guaranteeing specified constraints on computational resources and/or error tolerance is in general untenable, instead, a numerical analysis is performed. In the central and right panels of Fig.~\ref{fig:dist_ITE_GEO} the convergence of QITE with respect to ITE employing measure \eqref{eqn:speed} for different values of the domain parameter $D$ is studied. Furthermore the intermediate times distance between trajectories (central panel) and the final trajectory distance (right panel) are also compared. As can readily be seen the trajectory distance is greatly decreased when $D$ is high, as already anticipated in this section.

\section{Adaptive Compressed QITE (ACQ)}\label{sec:adapQITE}

\subsection{Motivation}

In the previous section, the geometric aspects of ITE were reviewed; particularly how the path followed by this evolution differs from geodesic paths as the dimension of Hilbert the space grows. An example of the latter is provided in Fig.~\ref{fig:dist_ITE_GEO}.
It was pointed out that only when the Hamiltonian of interest has spectral cardinality 2, the path that ITE and QITE generates is a geodesic path. When moving to Hamiltonians with higher spectral cardinality, one can employ measure \eqref{eq:traj_distance} to quantify how much ITE deviates from the geodesic path. Note that if the time step of QITE is sufficiently small, its path closely reproduces the one from ITE. These features are sketched in Fig.~\ref{fig:trajectories}.

Let us for a moment assume that the ITE path does not deviate much from a geodesic path. In that case, the QITE routine would generate unitaries whose directions do not change much from one step to another, and it would be redundant to keep computing new unitaries. Whence, it would be more efficient to reapply the same unitary multiple times, and in this case the resulting state would deviate a small bit after a few steps. At some point, the deviation from the intended ITE path would be so great that one would not be evolving towards the ground state, at which point one could repeat the same procedure: compute the QITE unitary and propagate it again until one deviates from the ITE path or stop if one gets close to the target state, but a criterion is needed to decide this. 

Since the energy along ITE always decreases, an energy increase of a path that attempts to reproduce ITE would be a signature of devitation from ITE. If the expectation value of the Hamiltonian along an iterated QITE unitary increases after some iteration, it would indicate a deviation from the gradient descent. 
This new method is depicted in Fig.~\ref{fig:trajectories} by a red path, where each red arrow represents an execution of the QITE routine and the dashed red line the iteration of the resulting unitary until energy increases. This method would suppose less classical overhead and also require fewer measurements, but it would require the application of more unitaries than with QITE. This problem is addressed in the following section.

The case of spectral cardinaliy 2 is a special case in which the geodesic interpretation of ITE is exact, as discussed
in Sec.~\ref{sec:geom} and Appendix \ref{sec:app_geometric}. For Hamiltonians with higher spectral cardinality, the ITE trajectory is
not generally geodesic; this is precisely why the trajectory measures
introduced in Secs.~\ref{sec:geom1} and \ref{sec:geom2} are needed. Furthermore, Appendix \ref{sec:app_Length} provides upper and lower bounds for the ITE-convergence time, suggesting a modest scaling of  $\mathcal{D}_{v}(S_1,S_2)$ in Eq.~\eqref{eqn:speed} with respect to $N$; thus providing a more generic analysis of the ACQ performance for the case of Hamiltonians of arbitrary spectral cardinality.  Finally, it is worth noting that the geodesic
intuition is used locally: ACQ follows a QITE descent direction only while the energy decreases, and recomputes the QITE generator once this direction ceases to provide further energy reduction.

\begin{figure}
    \centering
    \includegraphics[width=0.5\linewidth,trim=-10 -10 -10 -10]{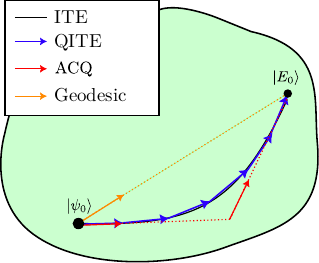}
    \caption{Geometric representation of ITE, QITE, and ACQ trajectories on the complex projective plane $\mathbb{CP}^N$. The initial state $\ket{\psi_0}$ evolves toward the ground state $\ket{E_0}$ via different paths. The black curve shows the ITE trajectory following gradient descent. The orange dashed line represents the geodesic (shortest path on $\mathbb{CP}^N$). Blue arrows show QITE discrete unitary steps, where each step requires computing a new unitary. Red arrows illustrate ACQ, which extends a single unitary propagation until energy increase signals deviation from the ITE path, thereby reducing the number of QITE routine executions while maintaining comparable accuracy.}
    %\caption{Sketch of the trajectories taken by different methods. Geodesic/unitary trajectories are represented by straight lines. The ITE trajectory given by the gradient descent equation is tightly reproduced by discrete time steps of QITE. Adaptive QITE extends the unitary evolution of a QITE step until the energy of the evolved state starts increasing, at which point a new QITE step is computed and propagated. }
    \label{fig:trajectories}
\end{figure}

\subsection{The Proposal}

With the above pictorial view in mind, the new Adaptive time Compressed QITE (ACQ) algorithm is now introduced; an algorithm which combines an adaptive time step with a merging of the resulting unitaries of the QITE routine.

The ACQ algorithm relies on the QITE routine explained in Sec.~\ref{sec:QITE} to obtain the unitaries generated by $\boldsymbol{\hat{A}}_{n,k}$ that approximate ITE (see Eq.~\eqref{eq:QITE_cond}) generated by all the Hamiltonian pieces in Eq.~\eqref{eq:HamPieces}. Instead of recomputing the unitary
\begin{equation}
\label{eq:QITEunitary}
    \boldsymbol{\hat{U}}_m=\prod_k e^{-i \Delta \tau\boldsymbol{\hat{A}}_{m,k} }~,
\end{equation}
that evolves QITE at every step, here it is proposed that the same unitary be reused (following the red dashed line in Fig.~\ref{fig:trajectories}) until the energy of the evolved state starts increasing. The ACQ energy steps are labelled with a superscript $m$ and an additional parameter $l$ to distinguish them from the QITE steps. A line search in $l$ is performed until the energy increases, namely,
\begin{equation}\label{eq:ACQ_Eincrease}
    E_{l+1}^{(m)}-E_{l}^{(m)} > 0~,
\end{equation}
where the energy $E_l^{(m)}$ is defined as
\begin{equation}
    E_l^{(m)} =\bra{\psi_{m-1}} (\boldsymbol{\hat{U}}_m^\dagger)^l \boldsymbol{\hat{H}} \boldsymbol{\hat{U}}_m^l \ket{\psi_{m-1}}~,
\end{equation}
the expectation value of the Hamiltonian of the previous state $\ket{\psi_{m-1}}$ evolved $l$ times with the unitary \eqref{eq:QITEunitary}. In which case the state of the sequence would be 
\begin{equation}
    \ket{\psi_{m}}=(\boldsymbol{\hat{U}}_m)^{l_m}\ket{\psi_{m-1}}~,
\end{equation}
where $l_m$ is the first $l$ that obeys Eq.~\eqref{eq:ACQ_Eincrease}.
This method would greatly reduce the number of times the QITE routine is applied, and, in a practical scenario, the number of mid-circuit measurements required to compute the expectation values needed for computing the next step $\boldsymbol{\hat{A}}_{m,k}$. This scenario is represented in the second circuit of Fig.~\ref{fig:ACQ_steps}. This method also reduces significantly the number of times the classical optimization is performed.

\begin{figure}
    \centering
    \adjustbox{valign=c}{\includegraphics[width=0.47\linewidth]{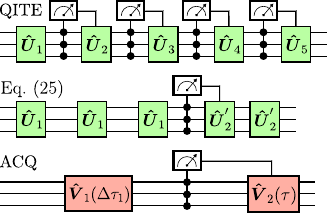}}
    \adjustbox{valign=c}{\includegraphics[width=0.51\linewidth]{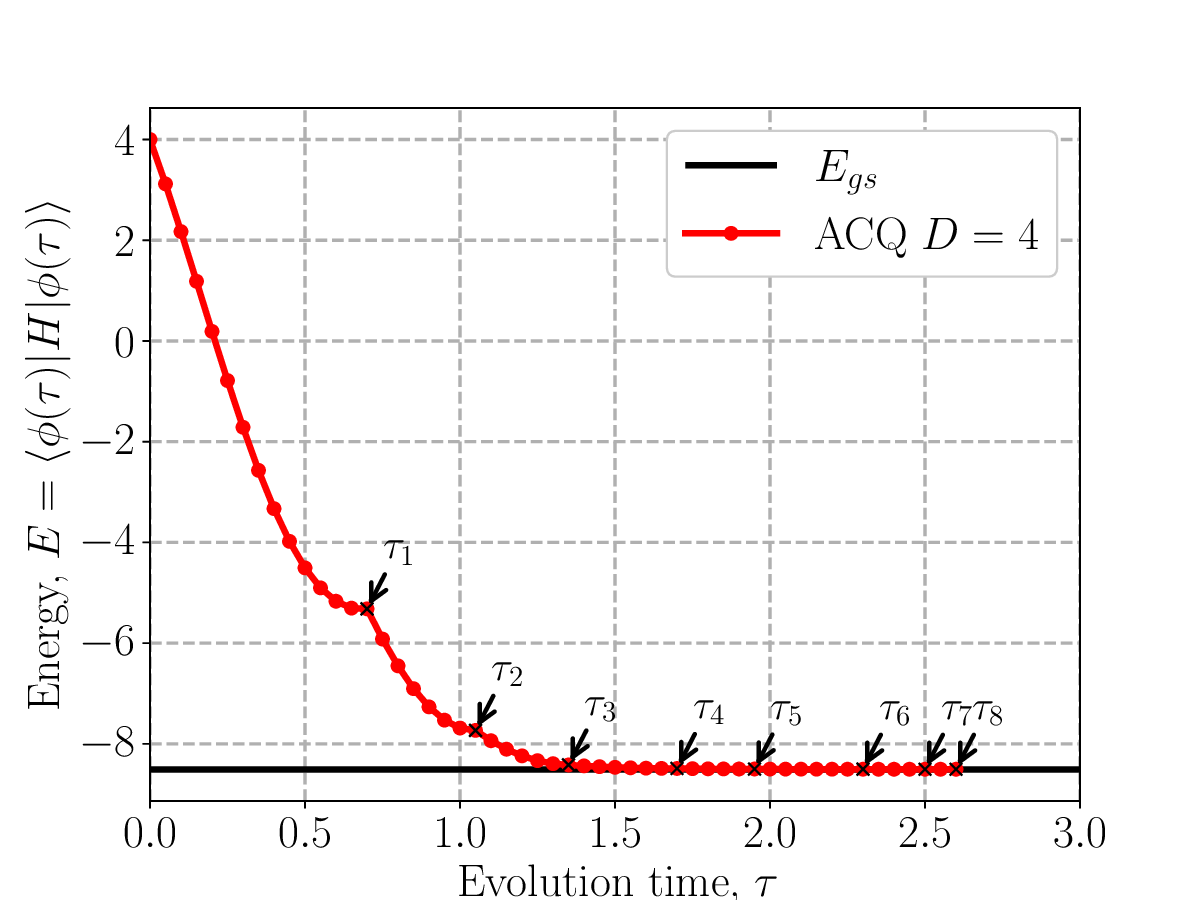}}
    \caption{\textbf{Left}: Circuit representation of QITE and ACQ. In the first circuit, measurements are performed after each unitary and the outcomes are used to create the following one. In the second circuit, the same unitary is reused multiple times until the energy increases, at which point a measurement is performed to obtain the following unitaries. In the last circuit, a compression of the unitaries of the previous circuit is performed showing both the reduction in depth and measurements of ACQ with respect to QITE. 
    \textbf{Right}: Adaptive time step calculation. The expectation value of the energy is plotted at finite times $\tau$ (sampled at time intervals $\Delta \tau =0.05$) of the state $\ket{\phi(\tau)}=\boldsymbol{\hat V}_n(\tau-\tau_{n-1})\ket{\psi_{n-1}}$ with $\ket{\psi_n}=\prod_{i=1}^n \boldsymbol{\hat V}_i(\Delta \tau_i)\ket{\psi_0}$ and $\tau_n > \tau > \tau_{n-1}$. The times $\tau_n$ at which the expectation value of the energy would start to increase, marked with an arrow, define the steps of ACQ. The inner workings of ACQ can be seen by a stalling of energy before $\tau_n$ and subsequent stepper descent in the energy after a new QITE unitary calculation. The adaptive time step $\Delta \tau_n=\tau_n-\tau_{n-1}$ is given by the time difference of these points. Here, the same Hamiltonian as in Fig.~\ref{fig:dist_ITE_GEO} is employed for a chain of $N=8$ qubits split into pieces of locality $K=2$ and a domain size of $D=4$ for the QITE unitaries.}
    \label{fig:ACQ_steps}
\end{figure}

\begin{table}[t]
    \noindent\rule{\linewidth}{0.4pt}
    \vspace{-0.4cm}
    \begin{algorithmic}[1]
      \STATE \textbf{Input:} Hamiltonian $H=\sum_k h_k$; initial state $\ket{\psi_0}$; domain size $D$; time increment $\Delta\tau$
      \STATE \textbf{Output:} States $\{\ket{\psi_m}\}$ and energies $\{E_m\}$
      
      \STATE Initialize $m \gets 0$
      \STATE $E_m \gets \langle \psi_m | H | \psi_m \rangle$
      \REPEAT
        \STATE $m \gets m + 1$
        \STATE $\{\{A_{m,k}\}, E_{\mathrm{next}}^{\mathrm{QITE}}\}
               \gets \mathrm{QITE\ on\ }(D,\{h_k\},\ket{\psi_{m-1}})$ \hfill \texttt{\# $E_\textrm{next}^\textrm{QITE}$ is the energy after QITE}
        \STATE $V_m(\tau) \gets \exp\!\big(-i\, \tau \sum_k A_{m,k}\big)$
        \STATE $\Delta\tau_m \gets 0$
        \REPEAT
          \STATE $\Delta\tau_m \gets \Delta \tau_m + \Delta\tau$
          \STATE $\ket{\psi_m} \gets V_m(\Delta \tau_m)\ket{\psi_{m-1}}$
          \STATE $E_m \gets \langle \psi_m | H | \psi_m \rangle$
          \STATE $E_{\mathrm{next}} \gets
            \langle \psi_m |
            V_m^\dagger(\Delta\tau)\, H\, V_m(\Delta\tau)
            | \psi_m \rangle$
        \UNTIL{$E_{\mathrm{next}} > E_m$} \hfill \texttt{\# This loop establishes the adaptive time step} 
        \STATE Store $\{\{A_{m,k}\},E_m,\Delta\tau_m\}$ 
      \UNTIL{$E_{\mathrm{next}}^{\mathrm{QITE}} > E_{m-1}$}
      \STATE \textbf{Return} $\{A_{m,k}\}, \{E_m\},\{\Delta \tau_m\}$ \hfill  \texttt{\# Discard last pair where energy increased}
    \end{algorithmic}
    \rule{\linewidth}{0.4pt}
    \caption{Pseudocode for the ACQ algorithm. }
    \label{alg:ACQ}
\end{table}

However, this method would increase the depth of the circuit that implements it. For that reason, the same techniques used in cQITE \cite{Nishi2021,Gomes2020} are employed, where a reverse Suzuki-Trotter decomposition is applied to the matrices that make up $\boldsymbol{\hat{U}}_m$ in Eq.~\eqref{eq:QITEunitary}. Moreover, this reverse decomposition allows us to define the following 1-parameter unitary matrix group
\begin{equation}\label{eq:CompressedU}
    \boldsymbol{\hat{V}}_m(\tau) = e^{-i \tau \sum_{k} \boldsymbol{\hat{A}}_{m,k} }~,
\end{equation}
%dir que simplement es un pas, un procediment per a poder fer l'algoritme, pero que després realment aplicarem exponencials en blocs de D a la QITE modulo error de trotter. 
that approximately reproduces $\boldsymbol{
\hat{U}}_m^l$ at times $\tau=l\Delta \tau$. This unitary allows to extend the unitary evolution produced by QITE to arbitrary times for a fixed-depth circuit, where different times only amount to a reparameterization of the gates composing $\boldsymbol{\hat{V}}_m(\tau)$. With this new definition the evolution of the previous state becomes
\begin{equation}
    \ket{\phi(\tau)}=\boldsymbol{\hat V}_m(\tau-\tau_{m-1})\ket{\psi_{m-1}}~,
\end{equation}
with $\tau_m > \tau > \tau_{m-1}$ where $\tau_m$ is the time at which $\ket{\psi_m}$ was found. One would start with $m=0$, $\ket{\psi_0}$ and $\tau_0=0$ and iteratively find all the subsequent times $\tau_m$. To find these $\tau_m$, the expectation value of the Hamiltonian is defined as
\begin{equation}
    E(\tau) = \bra{\phi(\tau)}  \boldsymbol{\hat{H}} \ket{\phi(\tau)}~,
\end{equation}
and \eqref{eq:ACQ_Eincrease} is conditioned to find the adaptive time step $\Delta \tau_m$ and the following step state $\ket{\psi_m}$ when 
\begin{equation}
    \frac{d}{dt}E (\tau) \Bigg|_{\tau=\tau_m} = 0~,
\end{equation}
where $\Delta \tau_m=\tau_m-\tau_{m-1}$ is the time difference between steps  and $\tau_m$ is the time at which the energy would start to increase. Since it is known that (Q)ITE initially follows a gradient descent, the extremal point $\tau_m$ will always be a minimum of the energy. 
The described procedure is iterative and requires finding previous steps to compute the following one as
\begin{equation}
    \ket{\psi_m}=\boldsymbol{\hat V}_m(\Delta \tau_m)\ket{\psi_{m-1}}~.
\end{equation}
Finally, the evolved state steps are
\begin{equation}
    \ket{\psi_m} = \prod_{j=1}^m \boldsymbol{\hat{V}}_j(\Delta \tau_j)\ket{\psi_0}~,
\end{equation}
where the adaptive time steps $\Delta \tau_j$ need to be found iteratively.
The ACQ evolution resource efficiency in terms of depth and mid-circuit measurements is represented in the last circuit of Fig.~\ref{fig:ACQ_steps}.
In the right panel of Fig.~\ref{fig:ACQ_steps} it is shown how the adaptive time step is obtained through an evolution  where the energy $E(t)$ is probed at finite time increments $\Delta \tau$.

The fixed increment $\Delta\tau$ used in this line search should be
distinguished from the adaptive propagation time $\Delta\tau_m$. The former is
the sampling resolution with which the energy profile along
$\hat{\boldsymbol{V}}_m(\tau)$ is probed, while the latter is the physical time
interval selected adaptively by the algorithm. Once the local QITE generator has
been computed, $\Delta\tau$ primarily controls the trade-off between
line-search resolution and measurement cost: smaller values locate the minimum
along $\hat{\boldsymbol{V}}_m(\tau)$ more accurately, whereas larger values
reduce the number of energy evaluations at the price of a coarser stopping
point. The usual QITE requirement that $\Delta\tau$ be sufficiently small to
construct a reliable local generator still applies. While the number of measurement required for computing a QITE unitary is exponential in the truncation domain size $D$, in general $N_K4^D$, where $N_K$ is the number of Hamiltonian pieces in Eq.~\eqref{eq:HamPieces}, the number of measurements required to evaluate the energy is at most linear with the system size $N$. For instance, for the TFIM model in Eq.~\eqref{eq:TFIM}, each evaluation of the energy only requires two basis measurements (one in the $Z$ basis and another in the $X$ basis)\footnote{These measurements have some sampling overhead required to reduce the shot noise.}. 

On the other hand, if one does not use any truncation in the domain $D$ of the unitaries of QITE, the ACQ algorithm (and QITE) would always decrease the energy, but the cost would remain exponential in the number of spins $N$ of the Hamiltonian. If a truncation of $D$ is used, both QITE and ACQ are limited and, at some point, the evolved state cannot get closer to the ground state, or decrease the energy. Below, the following stopping criterion is used for both algorithms. We proceed the algorithm until a point where the obtained unitary from QITE $\boldsymbol{\hat{U}}_n$ or the compressed unitary $\boldsymbol{\hat{V}}_m(\tau)$ generate a state with higher energy than the previous one, when this is achieved the algorithm is stopped. The detailed steps of ACQ are summarized in the Algorithm of Table \ref{alg:ACQ}, where the novelty of ACQ comes from checking if the energy of the following state increases in lines 15 and 17, and the compression of unitaries is performed in line 8.

\begin{comment}
    The stopping criterion E_next > E_n in line 12 ensures we halt the line search when
continuing to propagate the same unitary would increase the energy, indicating departure
from the ITE trajectory. Since ITE guarantees monotonic energy decrease (Eq. (21)), any
energy increase signals that the current unitary V_n(t) is no longer providing a good
approximation to the ITE dynamics. At this point, we compute a new QITE unitary.
Similarly, the outer stopping criterion in line 14 halts the entire algorithm when even a
newly computed QITE unitary would increase energy, indicating convergence or limitation
of the truncation scheme.
\end{comment}

It is worth mentioning that ACQ bears some resemblance to the so-called Boosted Imaginary Time Evolution (BITE) method introduced in \cite{boost_MPS}. There, similar geometric ideas were explored for matrix product states to reduce the number of times time-evolving block-decimation (TEBD) is applied on matrix product states, which is the costly part of the algorithm; in ACQ the goal is to reduce the amount iterations used to approximate ITE by unitaries. For ACQ, the approach involves implementing a line-search coupled with a compression scheme described above. These techniques mirror part 2 of the algorithm presented in Part B of \cite{boost_MPS} and adjustments of the bond-dimension\footnote{Analogous to the domain size $D$.} $\chi$ respectively. Although the methodology is different, the goal in both methodologies is to reduce the amount of approximate-ITE evolutions. In a sense, ACQ could be considered a sort of Boosted QITE (BQITE) since a classical version of ACQ, i.e., just an adaptive ITE, could be approximated by this BITE and vice-versa.

\section{Results}\label{sec:results}

\subsection{Performance comparison with QITE}

In Fig.~\ref{fig:fidelities_ACQ} two examples of the advantageousness of the ACQ method are exhibited. Once again, the TFIM presented in Eq.~\eqref{eq:TFIM} is used with parameters that make the ground state lie in the disordered phase where the ground state is non-degenerate ($J = 0.5$ and $h = 1$) and also the case of an ordered phase ($J=1$ and $h=0.5$) where the ground state is quasi-degenerate\footnote{It is degenerate in the thermodynamic limit. For a finite chain, the spectral gap is relatively small.}. It can be seen that in the disordered phase any algorithm can reach fidelities close to unity, but in the ordered phase, where the spectral gap is relatively small, both methods struggle to get good fidelities. This is a limitation of the truncation scheme of QITE, that is also present in ACQ. While increasing $D$ can help to reach better fidelities, the cost of the procedure is exponentially increased. This reason is what motivated the search of truncation schemes that are specially tailored for certain problems \cite{Alam2023,PhysRevA.109.052430,MolecularImprov}.
%For all numerical experiments, the initial state is prepared as $|\psi_{0}\rangle = |+\rangle^{\otimes^{N}} = (|0\rangle + |1\rangle)^{\otimes ^{N}} / 2^{N/2}$, which has uniform overlap with all computational basis states and ensures non-zero overlap with the ground state as required by ITE convergence. 
Both QITE and ACQ approximate the dynamics generated by ITE with a unitary at each step, but ACQ requires much fewer iterations of this procedure, which implies a partial state tomography and the resolution of the optimization procedure of Eq.~\eqref{eq:QITE_cond}. Gate count cost for a step of QITE and ACQ is comparable; a topic that will be further discussed in a following subsection.

\begin{figure}[t]
    \centering
    \includegraphics[width=\linewidth]{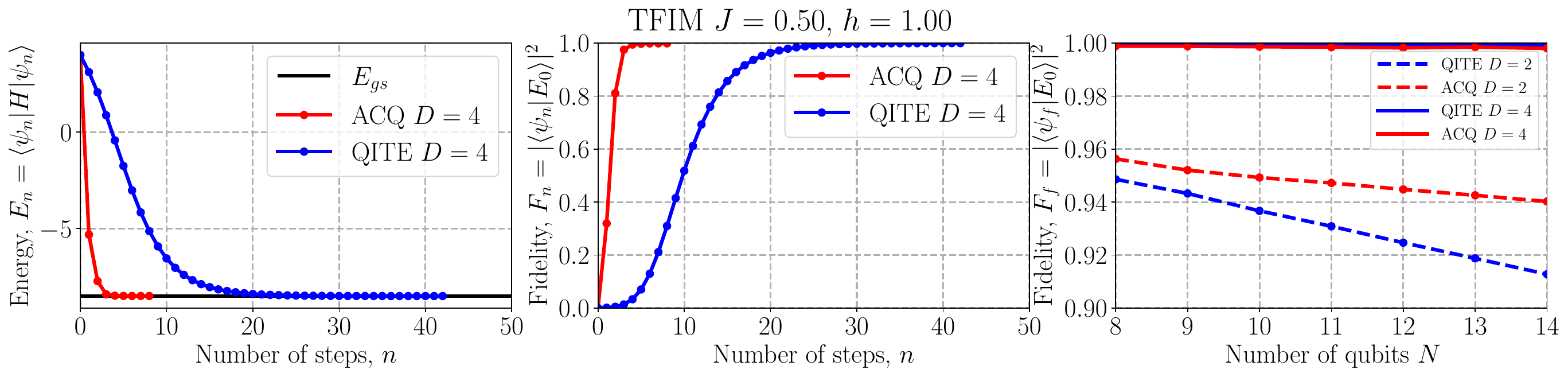}
    \includegraphics[width=\linewidth]{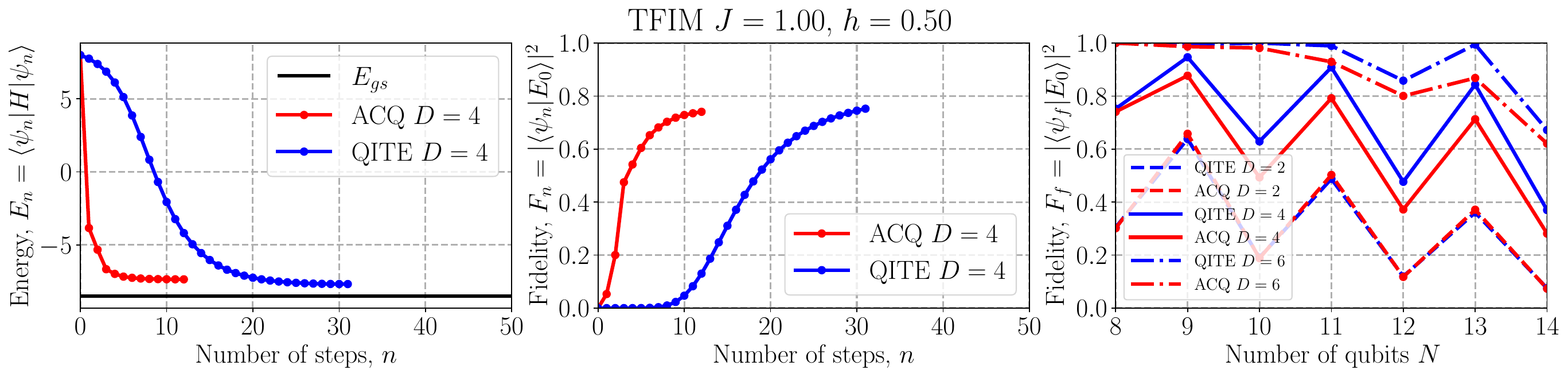}
    \caption{Comparison of QITE and ACQ for the TFIM with a ground state in the disordered phase ($J = 0.5$ and $h = 1$) in the top row, and a ground state in the ordered phase ($J = 1$ and $h = 0.5$) in the bottom row. \textbf{Left}: Energy evolution vs. iteration number for a system of $N = 8$ qubits. The black horizontal line indicates the exact ground state energy. Blue lines represent the evolution generated by regular QITE, and the solid red line the ACQ evolution. 
    ACQ requires way less QITE optimizations compared to standard QITE.  \textbf{Middle}: Fidelity evolution $F = |\langle E_0|\psi_n\rangle|^2$ showing that ACQ achieves comparable final fidelity with fewer QITE calls. \textbf{Right}: Maximum achieved fidelity vs. system size for increasing $N$. Both methods show similar fidelity degradation with increasing N, indicating this is a limitation of the domain truncation $D$ rather than the ACQ approach. Dashed lines: $D = 2$, solid lines: $D = 4$, dot dashed lines: $D=6$. In all cases, the time step considered to compute the QITE steps and optimization is $\Delta \tau=0.05$ and the initial state $\ket{\psi_0}=\ket{0}^{\otimes N}$.}
    \label{fig:fidelities_ACQ}
\end{figure}

In the right column of Fig.~\ref{fig:fidelities_ACQ} the accuracy of both methods is explored by comparing the maximum fidelity that can be reached with each method; this is done by varying system size and the truncation strategies. In the disordered phase, it can be seen that for a domain size of $D=2$ fidelities already reach a high value close to 0.95 while for $D=4$ they become closer to unity for both methods. Note that when the system size increases, the maximum fidelity decreases. 
This degradation is observed for all values of $D$, indicating it is a fundamental limitation of the domain truncation scheme rather than a deficiency of ACQ specifically. The oscillation of the maximum fidelity reached in the ordered phase is directly related with the spectral gap being dependent on the parity of $N$ \cite{sym14050996}.

\begin{figure}
    \centering
    \includegraphics[width=\linewidth]{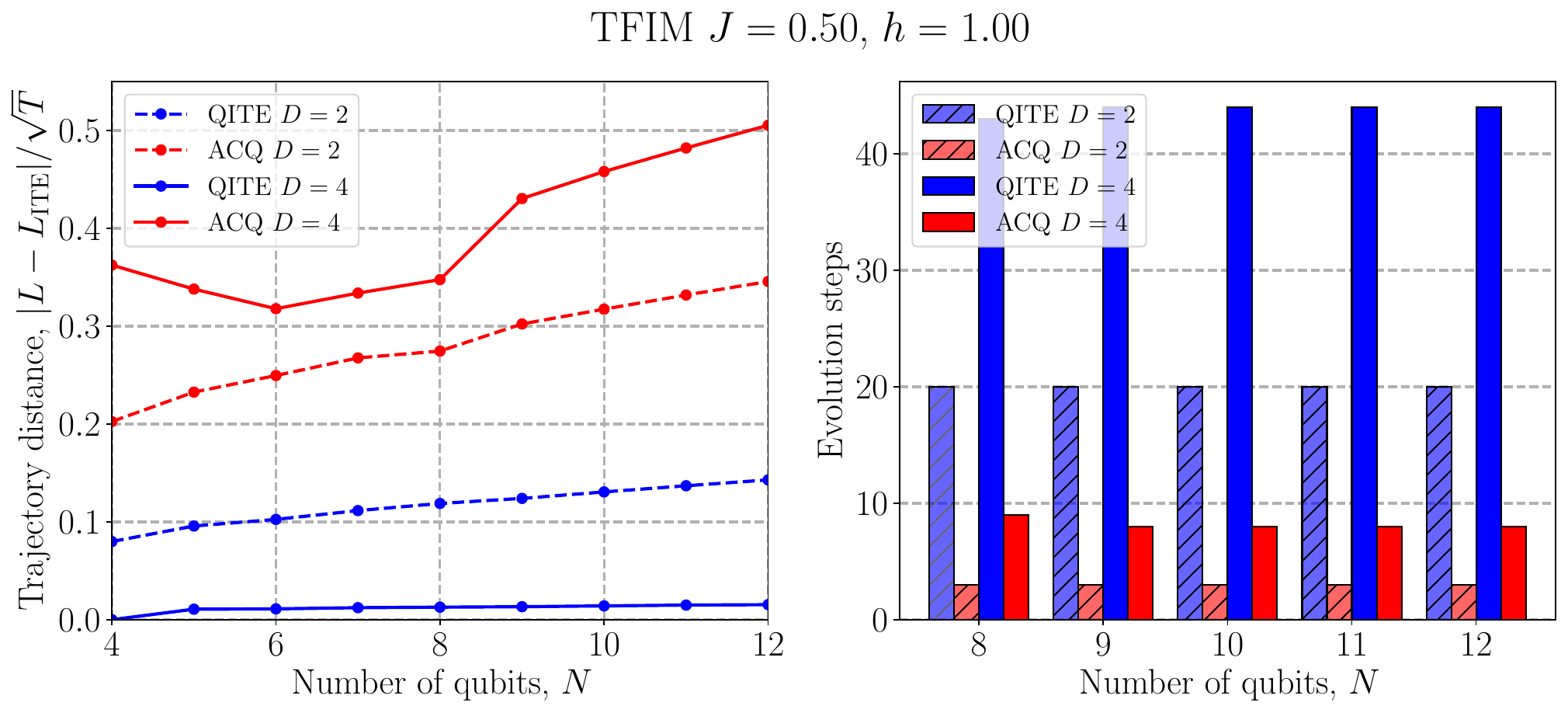}
    \caption{\textbf{Left}: Trajectory distance from QITE (in blue) and ACQ (in red) to ITE for different values of qubits $N$ and domain size $D$. Dashed lines: $D=2$, solid lines: $D=4$. \textbf{Right}: Number of steps taken for each algorithm to reach the minimum in energy for different number of qubits $N$ and domain size $D$. Blue bars correspond to QITE and red bars to ACQ. Hatched bars: $D=2$ , solid bars: $D=4$. In both case the TFIM has been considered in the disordered phase ($J=0.5$ and $h=1$), with the same parameters and initial states as in Fig.~\ref{fig:fidelities_ACQ}.}
    \label{fig:ACQ_dist_cost}
\end{figure}

The difference between the disordered and ordered phases is consistent with
the role of the locality parameter $D$ discussed in Sec.~\ref{geosec}. In the
disordered phase, correlations remain relatively short-ranged, so a modest
finite domain can already capture the relevant local tangent directions with
good accuracy. By contrast, in the ordered phase, longer-range correlations
become more important, and the same value of $D$ may no longer resolve the
effective imaginary-time direction sufficiently well. In this regime, the
finite-domain approximation neglects contributions from operators with support
larger than $D$, which can have non-negligible weight when the effective
imaginary-time direction is no longer well approximated by a quasi-local
generator.

This should be interpreted as a limitation inherited from finite-domain QITE,
rather than as a limitation introduced by the adaptive compression step. Indeed,
ACQ constructs each adaptive direction from the same QITE generator used in the
standard algorithm. Therefore, if the domain size is too small to capture the
correlations relevant to the imaginary-time tangent direction, both QITE and
ACQ are affected by the same truncation error. Conversely, when $D$ is large
enough to provide a reliable local descent direction, ACQ can exploit this
direction over a longer adaptive interval before recomputing the generator.

This caveat is also consistent with the discussion of double-bracket-type flows:
locality of the Hamiltonian alone does not guarantee that the corresponding
effective imaginary-time direction remains local at all times~\cite{Matt}.

It can be seen that the fidelities reached by ACQ are very similar to the ones obtained with regular QITE. One can already see that the sacrifice in accuracy is justified by the reduction in the cost pertaining to circuit depth and run-time. In a following subsection, among other things, $\mathcal{D}_v(S_{1}, S_{2})$ Eq.~\eqref{eqn:speed} will be used to analyse the discrepancy between QITE and ACQ in a more rigorous manner. Similar explorations to the one made here were performed for the Heisenberg model presented below, and the Cluster-Ising model, and are available in the repository associated with this work.

\subsection{Resource cost analysis}

So far, the QITE step $\boldsymbol{\hat U}_n$ in Eq.~\eqref{eq:QITEunitary} and the ACQ step $\boldsymbol{\hat V}_n(\Delta \tau_n)$ in Eq.~\eqref{eq:CompressedU} have been considered on an equal footing in terms of cost. It has been claimed that, in the cases of Fig.~\ref{fig:fidelities_ACQ}, ACQ reached comparable fidelities as QITE with much fewer steps, that is, fewer midcircuit measurement required to compute the following step unitaries. In Appendix \ref{sec:app_transpilation} arguments that justify that one step of QITE and ACQ also accrue the same cost in terms of number of gates are given. In the end, the cost comparison both in terms of gates and required measurements is equivalent for the same number of steps of both methods. For that reason, in the right panel of Fig.~\ref{fig:ACQ_dist_cost}, the number of required steps to reach the minimum of energy both for QITE and ACQ with comparable fidelities are compared. It can be observed that the number of ACQ steps is consistently smaller than for QITE for a wide range of system sizes and different values of the truncation parameter $D$.

\subsection{Comparison with DB-QITE}

The ACQ algorithm presented here inherits the hybrid nature of the QITE algorithm, i.e., unitary evolution which is constructed from a classical optimization to build the unitary that more closely approximates the ITE, after a partial tomography through measurements between steps. The novelty of the DB-QITE algorithm introduced in \cite{glu} is its pure quantum nature, as it does not require mid-circuit measurements. 
In this section, we compare the performance of both algorithms for two models: the TFIM already introduced in Eq.~\eqref{eq:TFIM}, and the antiferromagnetic Heisenberg model defined by the Hamiltonian
\begin{equation} 
    H = \sum_{i=1}^{N-1} \boldsymbol{\hat \sigma}^x_j\boldsymbol{\hat \sigma}^x_{j+1} +\boldsymbol{\hat \sigma}^y_j\boldsymbol{\hat \sigma}^y_{j+1} + \boldsymbol{\hat \sigma}^z_j\boldsymbol{\hat \sigma}^z_{j+1}~,
\end{equation}
(where now this system has open boundary conditions) which is the model explored in that reference, in order to make a direct comparison. Notice that this model possesses a high degree of symmetry, and thus it allows for an efficient preparation of a warm start. In other words, results from this model cannot be extrapolated to other models, as it will be made clear for the TFIM below. In particular, a Hamiltonian Variational Ansatz (HVA) \cite{mcardle2019variational,PRXQuantum.1.020319}, 
and a state that is a pair-wise product of singlet states, was used in \cite{glu}. As an alternative, we also employ an initial state with alternating spins $\ket{\psi_0}=\ket{01}^{\otimes N/2}$ that has, in opposition to previous states, a small overlap with the ground state.
For the TFIM, we consider the same initial state as before  $\ket{\psi_0}=\ket{0}^{\otimes N}$ and a HVA state; more details of how this state is prepared for the TFIM are given in Appendix~\ref{sec:app_transpilation}.
We have employed the computational methods provided in the repository of \cite{glu} to compute the DB-QITE dynamics, and used the transpilation methods therein to evaluate the resource cost of either ACQ and DB-QITE presented below. See  Appendix~\ref{sec:app_transpilation} for further details.

\begin{figure}[t]
    \centering
    \includegraphics[width=\linewidth]{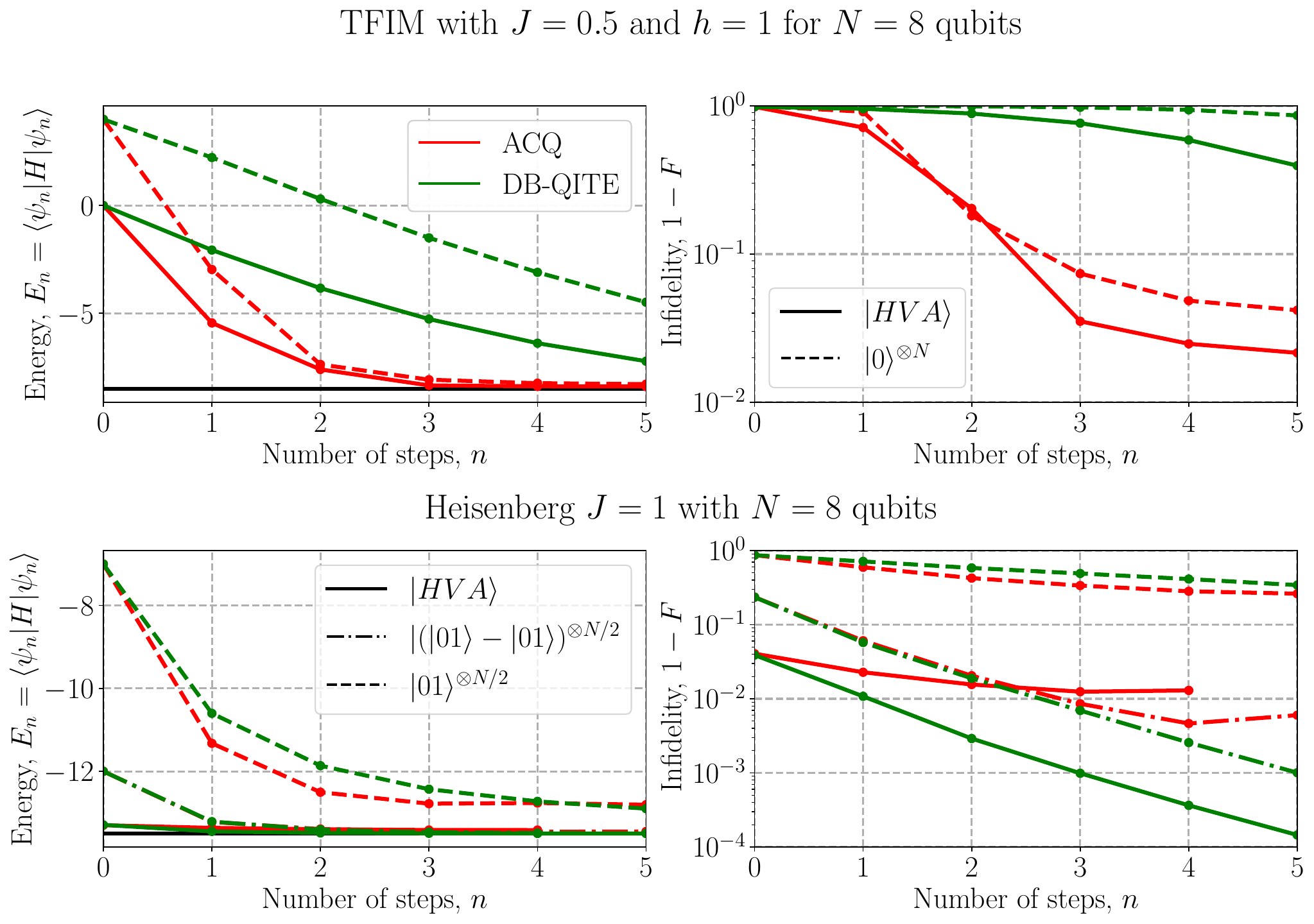}
    \caption{Comparison of the energy decrease (left) and fidelity increase (right) between ACQ and DB-QITE for the TFIM (top row) and the Heisenberg model (lower row). The red (green) colour is used to refer to ACQ (DB-QITE). Different stroke lines represent different initial state evolutions captured in the legend of each model. Since fidelities reach values close to 1, we plot the infidelity instead in logarithmic scale, to easily differentiate the maximum accuracy achieved (lower is better).} 
    \label{fig:comp_DBF}
\end{figure}

The results of ACQ come from the transpiled version of the algorithm discussed in Appendix~\ref{sec:app_transpilation} since the results of the DB-QITE are also the ones coming from the transpiled circuit; both methods include transpilation, necessary for counting gate resources, and also include the corresponding error generated by it.

Figure~\ref{fig:comp_DBF} shows the energy reduction, and the infidelity decrease obtained for both methods within all the described scenarios.
It can be observed that the convergence to the ground state energy of ACQ for the TFIM model is faster in the number of steps, while in the Heisenberg model the speed of convergence is comparable in both cases. If we focus on infidelity decrease, in the TFIM model the decrease is stepper for ACQ. In the Heisenberg model, where higher fidelities are reached, we can see that ACQ fails to match  the rate the DB-QITE exhibits after fidelities over $99\%$. This fact is due to the truncation strategy available to ACQ (and QITE), where the size of the unitaries is restricted to be $D$-local. If higher fidelities were required, the truncation parameter $D$ could be increased, at the expense of increasing the resource cost. 

Table~\ref{tab:resourcecost_ACQDB} gathers the numerical values of the fidelities and gate resources in terms of \texttt{u3} and \texttt{cx} gates for the number of steps of each method, for two particular initial states of each problem Hamiltonian. We have also included the case $D=2$, not shown in Fig.~\ref{fig:comp_DBF}, to stress out the relevance of choosing an appropriate truncation scheme. On the one hand, thanks to the adaptive nature of ACQ, the fidelity quickly jumps to high values in the early steps, while the later ones let ACQ reach the maximum fidelity allowed by the truncation. It can be seen that the number of gates increases in a linear fashion per step. On the other hand, DB-QITE has a more modest fidelity increase, and due to the recursive construction of its unitaries, the number of gates increases exponentially (by approximately a factor 3)\footnote{The following step unitary is composed by 3 of the previous step unitaries, and also some additional operations in between, see \cite{glu}.} per step. Since DB-QITE does not include a truncation scheme, the fidelity increase does not appear to be bounded, as seen in Fig.~\ref{fig:comp_DBF} where the infidelity for the initial HVA state reaches values of $10^{-4}$. We also included the number of Pauli string measurements required for each step of ACQ. 
To summarize, even if ACQ needs to perform intermediate circuit measurements, which is optimized by the adaptive  strategy of ACQ, it represents an exponential reduction of resources as the number of steps is increased when compared with DB-QITE.

\subsection{Geometric path comparison}

The geometric analysis in Section \ref{sec:geom} provides insight into QITE's performance in reproducing the cooling process of ITE. The left panel of Fig. \ref{fig:dist_ITE_GEO} shows that ITE trajectories deviate more from geodesics as $N$ increases suggesting that QITE, which is composed of a concatenation of unitaries, will also depart from ITE as $N$ increases. In the left panel of Fig.~\ref{fig:ACQ_dist_cost} one can see that this is indeed the case, and, as already pointed in Fig.~\ref{fig:dist_ITE_GEO} as $D$ is increased it can be seen here that the distance to ITE is reduced. ACQ on the other side is not trying to reproduce ITE, and as the adaptive time step $\Delta \tau_n$ is coarse in comparison to QITE's step (see Fig.~\ref{fig:trajectories} for reference), it is expected that the ACQ path deviates from ITE more than QITE. While increasing $D$ reduced the distance to ITE for QITE, this is not the case for ACQ, where increasing $D$ helps getting closer to the ground state and hence travel more distance. 

It is worth to point out that, even though ACQ deviates significantly from ITE, the right column of Fig.~\ref{fig:fidelities_ACQ} demonstrates that ACQ still achieves comparable fidelity to QITE even at $N = 14$. This indicates that exact geodesic following is not necessary for good performance; rather, the key is that locally, ITE is well-approximated by geodesic segments, which ACQ exploits through its line search strategy.

\begin{table}[t]
    \centering
\begin{subtable}{0.55\textwidth}
\centering
\resizebox{\textwidth}{!}{%
\begin{tabular}{r|rrr|rrr|rrr}
\toprule
$n$ & \multicolumn{3}{c}{ACQ $D=2$} & \multicolumn{3}{c}{ACQ $D=4$} & \multicolumn{3}{c}{DB-QITE} \\
 & $F$ & \texttt{u3} & \texttt{cz} & $F$ & \texttt{u3} & \texttt{cz} & $F$ & \texttt{u3} & \texttt{cz} \\
\midrule
0&0 & 0 & 0 & 0 & 0 & 0 & 0 & 0 & 0 \\
1&0.61 & 40 & 16 & 0.09 & 334 & 392 & 0 & 209 & 130 \\
2&0.92 & 72 & 32 & 0.82 & 1958 & 2224 & 0.01 & 733 & 462 \\
3&0.93 & 104 & 48 & 0.93 & 3587 & 4056 & 0.02 & 2217 & 1404 \\
4&0.93 & 136 & 64 & 0.95 & 5226 & 5885 & 0.06 & 6258 & 3998 \\
5& &  &  & 0.96 & 6866 & 7717 & 0.14 & 18415 & 11832 \\
\midrule
$N_M$ & \multicolumn{3}{c}{36} & \multicolumn{3}{|c|}{666} & \multicolumn{3}{c}{0}\\ 
\bottomrule

\end{tabular}
}
    \caption{TFIM disordered phase with initial state $\ket{\psi_0}=\ket{0}^{\otimes N}$}
    \label{tab:TFIM_ACQ_DB}
   \end{subtable} 
   \begin{subtable}{0.43\textwidth}
    \centering
   \resizebox{\textwidth}{!}{% 
\begin{tabular}{r|rrr|rrr}
\toprule
$n$ & \multicolumn{3}{c}{ACQ $D=4$} & \multicolumn{3}{c}{DB-QITE} \\
 & $F$ & \texttt{u3} & \texttt{cz} & $F$ & \texttt{u3} & \texttt{cz}   \\
\midrule
0 & 0.765 & 12 & 4 & 0.765 & 12 & 4 \\
1 & 0.939 & 171 & 178 & 0.943 & 309 & 162 \\
2 & 0.98 & 503 & 551 & 0.981 & 1128 & 609 \\
3 & 0.991 & 828 & 930 & 0.993 & 3585 & 1950 \\
4 & 0.995 & 1144 & 1301 & 0.997 & 10938 & 5964 \\
5 & 0.994 & 1477 & 1670 & 0.999 & 32763 & 17889 \\
\midrule
$N_M$  & \multicolumn{3}{c|}{666} & \multicolumn{3}{c}{0}\\ 
\bottomrule
\end{tabular}
}
    \caption{Heisenberg model initial singlet state.}
    \label{tab:Heisenberg_ACQ_DB}
\end{subtable}
\caption{Table containing the resource cost comparison between the ACQ algorithm proposed here and the DB-QITE algorithm proposed in \cite{glu}. We include a column for the cost of single qubit gates \texttt{u3}, one for two qubit gates \texttt{cx} and another for the fidelity $F$ reached at each step $n$. Step 0 counts how many gates are required to prepare the initial state, and the fidelity of the initial state with respect to the ground state. The last row indicates the number of measurements $N_M$ required per time step for each method. We include two tables for the comparison with different models and different initial states.}
\label{tab:resourcecost_ACQDB}
\end{table}

\section{Conclusions}\label{sec:conclusion}

The main contribution of this paper is the Adaptive time Compressed QITE (ACQ) algorithm. An algorithm developed for the implementation of imaginary time evolution (ITE) preparation of ground states on quantum hardware in a more resource-efficient way than the original proposal \cite{Motta2019DeterminingEA}. 
An intensive review of this method and alternatives in the literature that yield improvements was made in Section \ref{sec:intro}. 

In Section \ref{sec:geom} we reviewed the geometrical properties of ITE present in the literature, we showed that QITE reproduces geodesic trajectories for the simple case corresponding to Hamiltonians with spectral cardinality 2 and introduced a measure to quantify the departure of any of these trajectories from the geodesic path; a measure which may also be used for studying the proximity of trajectories generated via any two of the popular iterations of ITE. It was observed that this measure for ITE increased for larger $N$. When comparing QITE with ITE, it was observed that as the domain size $D$ or the truncation procedure are improved, the measure reduced. 

Equipped with the geometric knowledge of ITE methods, arguments that justify the use of some approximations that help in reducing the resource cost of QITE, in terms of circuit depth and measurements, are given and the novel algorithm ACQ is introduced. The novelty of ACQ, comes from finding the adaptive time steps $\Delta \tau_n$ that yield a minimum in energy at each step, and compressing the unitaries to keep the gate count at bay. Wherefrom a significant reduction in cost is achieved with ACQ whilst maintaining fidelities comparable to those of QITE. The comparison with the Double Bracket QITE algorithm showed an advantageous gate count scaling for similar fidelities and sometimes with better performance in less structured models.

\ack{
This work was supported by the project PID2023-152724NA-I00, with funding from MCIU/AEI/10.13039/501100011033 and FSE+,
the Severo Ochoa Grant CEX2023-001292-S, Generalitat Valenciana grant CIPROM/2022/66, the Ministry of Economic
Affairs and Digital Transformation of the Spanish Government through
the QUANTUM ENIA project call - QUANTUM SPAIN project, and by the
European Union through the Recovery, Transformation and Resilience
Plan - NextGenerationEU within the framework of the Digital Spain
2026 Agenda, and by the CSIC Interdisciplinary Thematic Platform (PTI+)
on Quantum Technologies (PTI-QTEP+). This project has also received
funding from Horizon Europe EU projects MSCA-SE CaLIGOLA, Project
ID: 101086123, and MSCA-DN CaLiForNIA, Project ID: 101119552. 

T.P. acknowledges the support of the Generalitat Valenciana under grant CIPROM/2022/66, which facilitated research stays at the University of Valencia during 12–17 January and 14–20 July 2025 and also   acknowledge the Research Council of Finland for funding through Grant No. 359284/Finnish Quantum Flagship.

A.A.M would like to thank the COMCUANTICA/007 from the Generalitat Valenciana, Spain for supporting him during the development of this work.

R.G.L. is funded by grant CIACIF/2021/136 from Generalitat Valenciana.

M.A G-M acknowledges support from the Ministry for Digital Transformation and
of Civil Service of the Spanish Government through the QUANTUM ENIA project
call—Quantum Spain project, and by the European Union through the Recovery,
Transformation and Resilience Plan—NextGenerationEU within the framework of the
Digital Spain 2026 Agenda: also from Projects of MCIN with funding from Euro-
pean Union NextGenerationEU (PRTR-C17.I1) and by Generalitat Valenciana, with
reference 20220883 (PerovsQuTe) and COMCUANTICA/007 (QuanTwin).
}

\section*{Code Availability}

The code used to generate the results presented in this study is publicly available in an open-access repository. The source code is available at \url{https://github.com/dark-dryu/ACQ}

\appendix 

\section{Geometric Properties of ITE and geodesics in complex projective space}\label{sec:app_geometric}

%\subsection{ITE and QITE, Geodesics for the case of Rank-2 Hamiltonians}\label{sec:ITE_geodesics}

Let us restrict ourselves to the case of one spin, i.e., the Hilbert space of interest is $\mathbb{C}^2$. Without loss of generality, consider an arbitrary Hamiltonian $\boldsymbol{\hat{H}}$ with the spectral decomposition $\boldsymbol{\hat{H}}= E_{0}|E_{0}\rangle\langle E_{0}|+ E_{1}|E_{1}\rangle\langle E_{1}|$. Furthermore, consider an initial state $|\psi(0)\rangle$, such that $\langle E_{0}|\psi(0)\rangle\neq0$. Then, we have the following result. 
\begin{definition2}
\label{eqn:2disgeo}
Let $\boldsymbol{\hat{\rho}}(0)=|\psi(0)\rangle\langle \psi(0)|$, where we express the density operator in the Pauli basis
\begin{equation}
\rho(0) = \frac{1}{2}\Big(\mathbb{I} + \vec{r}(0)\cdot\boldsymbol{\hat{\vec{\sigma}}} \Big)~,
\end{equation}
with $\vec r(0)=(r_1,r_2,r_3)$, 
then, the antihermitian matrix $[\boldsymbol{\hat{\rho}}(0), \boldsymbol{\hat{H}}]$ generates the geodesic (one-parameter group) connecting $|\psi(0)\rangle$ and $|E_{0}\rangle$. Hence, the smallest arc between these two states may be parametrized as follows.
\begin{equation}
|\psi(\gamma)\rangle=e^{\gamma[\boldsymbol{\hat{\rho}}_{0}, \boldsymbol{\hat{H}}]}|\psi(0)\rangle~,
\end{equation}
where 
\begin{equation}
 0\leq\gamma \leq  \frac{2\arccos\Big(\sqrt{\frac{1+r_3}{2}}\Big) }{\omega\sqrt{(r_1^2+r_2^2)}}~,
\end{equation}
where $\omega=E_1-E_0$.  Here, $\gamma=0$ corresponds to the initial state and the upper bound corresponds to the ground state.
\end{definition2}
\begin{proof}
	
The density matrix of the initial state can be represented with the polarization vector $\vec{r}(0)$ as
\begin{equation}\label{eq:InitialBlochVector}
    \boldsymbol{\hat{\rho}}(0) =|\psi(0)\rangle\langle \psi(0)| = \frac{1}{2}\Big(\mathbb{I} +\vec{r}(0) \cdot \hat{\vec{\boldsymbol{\sigma}}} \Big)~,
\end{equation}
where $\vec r(0) = (r_1,r_2,r_3)$ with $\norm{\vec r(0)}=1$, and $\hat{\vec{\boldsymbol{\sigma}}}$ the Pauli vector.
%where we use the Bloch-vector decomposition for one-spin pure states presented in Appendix \ref{eqn:appgeo}. 
%Let $\boldsymbol{\hat{\rho}}:=|\psi(0)\rangle\langle \psi(0)|$, then 
The exponential in the unitary \eqref{eq:DB-ITE} can be computed for a general 1-qubit Hamiltonian expressed as $\boldsymbol{\hat{H}}=E_1\ketbra{E_1}+E_0\ketbra{E_0}$, so that 
\begin{equation}
    [\boldsymbol{\hat{\rho}}(0), \boldsymbol{\hat{H}}]  =\Big[\frac{1}{2}\Big(\mathbb{I} +\vec r(0)\cdot\boldsymbol{\hat{\vec{\sigma}}}\Big), \boldsymbol{\hat{H}}\Big] = i \frac{\omega}{2} \vec{n}\cdot\boldsymbol{\hat{\vec{\sigma}}}~,
\end{equation}
with $\vec{n}=(-r_2,r_1,0)=\vec r(0) \times\hat{z}$ and $\omega=E_1-E_0$. 
From elementary geometry it is known that this is the generator of a rotation around ${\vec{n}}$ whose respective rotations generate the great circle on the Bloch sphere which includes both $\hat{z}$ and ${\vec{r}}$. We can express the generated unitary as $U(\vec{n},s) =  e^{is\frac{\omega}{2}\boldsymbol{\vec{n}}\cdot \boldsymbol{\vec{\sigma}}}$
which represents a rotation about $\vec{n}$, with angle $\theta=\omega s \norm{\vec{n}}$, of polarization vectors, i.e.,
\begin{equation}
    U(\vec{n},s) \frac{1}{2}\Big(\mathbb{I} +\vec{r} \cdot \hat{\vec{\boldsymbol{\sigma}}} \Big) U^\dagger (\vec{n},s)= \frac{1}{2}\Big(\mathbb{I} +(R_{\hat{n}}(\theta)\vec{r}) \cdot \hat{\vec{\boldsymbol{\sigma}}} \Big)~,
\end{equation}
where $\hat{n}=\vec{n}/\norm{\vec{n}}$ is the normalized vector\footnote{We use the hat to refer to normalized vectors}.
This means that there exists an $s$ such that 
\begin{equation}\label{eqn:desired}
    \boldsymbol{\hat{\rho}}(s)=U(\vec{n},s) \boldsymbol{\hat{\rho}}(0) U^\dagger(\vec{n},s) = |E_{0}\rangle\langle E_{0}|=\boldsymbol{\hat{\rho}}_{gs}~.
\end{equation}
To find out what this $s$ is, we will use the Fubini-Study metric in Eq.~\eqref{eq:fubini_dist} to check that $d_{FS}(\ket{\psi(0)},\ket{E_0}) = d_{FS}(U^\dagger(\vec{n},s)\ket{E_0},\ket{E_0})$, so that 
\begin{align}
     \nonumber \arccos\big({\sqrt{\langle E_0|\boldsymbol{\hat{\rho}}(0)|E_0\rangle}}\big) &=
   \arccos\bigg( \sqrt{\langle E_{0}|U^{\dagger}(\vec{n},s)|E_{0}\rangle\langle E_{0}| U^(\vec{n},s) |E_{0}\rangle}\bigg)  \\
 &=\arccos\big( \big|\langle E_{0}|U(\vec{n},s)|E_{0}\rangle|\big) ~. \label{eqn:desired_distance}
 \end{align}
We can use the exponential of a Pauli vector to compute
\begin{equation}
    U(\vec{n},s) = \cos\left(\frac{\omega s}{2}\norm{n}\right) \mathbb{I} - i \sin \left(\frac{\omega s}{2}\norm{n}\right) \vec{n}\cdot\boldsymbol{\hat{\vec\sigma}}~,
\end{equation}
and noting that $\boldsymbol{\hat{\rho}}_{gs}=\ketbra{E_0}=\frac{1}{2}(\mathbb{I}+\boldsymbol{\hat{\sigma}}_z)$ we can compute $\langle E_{0}|U(\vec{n},s)|E_{0}\rangle$ as
\begin{equation}
     Tr\{\boldsymbol{\hat{\rho}}_{gs}U(\vec{n},s)\} = \cos\left(\frac{\omega s}{2}\norm{n}\right)~,
     %Tr\{\cos\left(\frac{\omega s}{2}\norm{n}\right) \mathbb{I} + i \sin \left(\frac{\omega s}{2}\norm{n}\right) \sigma_z\}~,
\end{equation}
which for $\vec n = (-r_2,r_1,0)$ implies in Eq.~\eqref{eqn:desired_distance} that
\begin{equation}
    \cos\left(\frac{\omega s}{2}\sqrt{r_1^2+r_2^2}\right) = \sqrt{\braket{E_0|\boldsymbol{\hat{\rho}}(0)|E_0}}=\sqrt{\frac{1+r_3}{2}}~,
\end{equation}
and that 
\begin{equation}\label{eqn:thetime}
    s=\frac{2\arccos(\sqrt{\frac{1+r_3}{2}})}{\omega\sqrt{r_1^2+r_2^2}}~.
\end{equation}
We can check in these equations that if $\boldsymbol{\hat{\rho}}(0)=\boldsymbol{\hat{\rho}}_{gs}$ we have $r_3=1$ and that $s=0$ (initial state is the same as the final one) and that if $r_3 \to -1$ (zero overlap with the ground state) then $s \to \infty$ in Eq.~\eqref{eqn:thetime}.
To verify that indeed Eq.~(\ref{eqn:thetime}) leads to Eq.~(\ref{eqn:desired}), let us directly compute $U^\dagger(\vec{n},s)\boldsymbol{\hat{\rho}}_{gs}U(\vec{n},s)$ and show that it is equivalent to our initial state $\boldsymbol{\hat{\rho}}(0)$. 
\begin{equation}
    U^\dagger(\vec{n},s)\boldsymbol{\hat{\rho}}_{gs}U(\vec{n},s) = \frac{1}{2}\Big(\mathbb{I} +(R_{\vec{n}}(\theta)\hat{z}) \cdot \hat{\vec{\boldsymbol{\sigma}}} \Big)~,
\end{equation}
where the rotation is given by Rodrigues' formula
\begin{equation}
    R_{\hat{n}} (\theta) \hat z= \hat z\cos\theta + (\hat{n} \times \hat{z})\sin\theta + \hat{n} (\hat{n}\cdot\hat{z}) (1-\cos\theta)~,
\end{equation}
with $\theta=\omega s \norm{\vec{n}} = 2\arccos(\sqrt{\frac{1+r_3}{2}})$ given in Eq.~\eqref{eqn:thetime} which simplifies to
\begin{equation}
\label{thecitationineed}
    R_{\hat{n}}(\theta) \hat z = \hat{z} \cos\theta +\frac{ r_1\hat{x} + r_2\hat{y}}{\sqrt{r_1^2+r_2^2}}\sin\theta~. 
\end{equation}
Finally, using the trigonometric identities
\begin{align}
    &\cos\left(2\arccos\left(\sqrt{\frac{1+r_3}{2}}\right)\right) =r_3~, \\
    &\sin\left(2\arccos\left(\sqrt{\frac{1+r_3}{2}}\right)\right) = \sqrt{1-r_3^2} ~,
\end{align}    
we can readily see that
\begin{equation}
    R_{\hat{n}}(\theta) \hat z = \vec{r}~,
\end{equation}
which is the Bloch vector of the initial state Eq.~\eqref{eq:InitialBlochVector}.

We have now established that the rotation $R_{\hat{n}}(\theta)$ applied to $\hat{z}$ yields the Bloch vector 
$\vec{r}$ of the initial state. To complete the proof, we need to determine the specific value of $s$ that 
transforms $|\psi(0)\rangle$ to the ground state $|E_{0}\rangle$. 
However, we have $R_{\hat{n}}(\theta)\hat{z} = \vec{r}$. Comparing with Eq.~(\ref{thecitationineed}), this requires: 
\begin{equation}
 \cos(\theta) = r_{3}
\end{equation}
\begin{equation}
\sin(\theta)= \sqrt{r_{1}^2 + r_{2}^2}.
\end{equation}
The second equation gives $\sin(\theta) = 1$ when $r_{1}^{2} + r_{2}^{2} \neq 0$. Since $\cos(\theta) = r_{3}$ and $\sin^[{2}(\theta) + 
\cos^{2}(\theta) = 1$, we verify: 
\begin{equation}
r_{3}^{2} + (r_{1}^{2} + r_{2}^{2}) = 1~,
\end{equation}
which is satisfied by the normalization of the polarization vector. The angle $\theta$ is therefore: 
\begin{equation}\theta = \arccos(r_{3}) = 2\arccos(\sqrt{(1+r_{3})/2})~,
\end{equation}
where we used the half-angle formula $\cos(\theta) = 2\cos^{2}(\theta/2) - 1$, which gives $\cos(\theta/2) = 
\sqrt{(1+r_{3})/2}$. 
From Eq.~(\ref{eqn:thetime}), we have $\theta = \omega s\sqrt{r_{1}^{2} + r_{2}^{2}}$. Substituting the expression for $\theta$ derived above: 
\begin{equation}
s = \frac{2\arccos(\sqrt{\frac{1+r_3}{2}})}{\omega\sqrt{r_1^2+r_2^2}}~.
\end{equation}
This establishes that the unitary $\boldsymbol{\hat{U}}(\vec{n}, s)$ with this value of $s$ 
transforms $|\psi(0)\rangle$ to $|E_{0}\rangle$, completing the proof. 

\end{proof}
Lemma \ref{eqn:2disgeo} demonstrates that with one iteration of the linearized version \eqref{eq:DB-ITE} of the gradient descent equation of interest, we are able to deduce the generator of the one-parameter group which transports the initial state $|\psi(0)\rangle$ to the ground state of $\boldsymbol{\hat{H}}$.

Owing to the relationship between DB-QITE and QITE, the reader probably intuits at this point, the trajectories traced out by QITE should approximately follow a geodesic given that ITE traces out a geodesic; i.e., the above lemma shows that for the case of one-qubit systems, one iteration of DB-QITE can produce the one-parameter group connecting the initial state to the ground state. Given that DB-QITE and QITE are algorithms that approximate each other \cite{glu}, it is expected that one should be able to produce such a unitary matrix after one iterative step for the case of QITE.  We provide a sketch of a proof for the latter in Appendix \ref{sec:1qbit_geo_QITE}. Therein we show that a single QITE iteration is enough to produce the generator of the geodesic connecting the initial state to the ground state.

\section{QITE moves along geodesics for the case of a single qubit}
\label{sec:1qbit_geo_QITE}

We begin by presenting a Lemma from \cite{arxisuzukisignaling}. 
\begin{definition2}
\label{eqn:suzukilemm}
Let $\boldsymbol{\hat{H}}$ be a Hermitian matrix and let $\alpha$ be a real number. Then, given an input state $\boldsymbol{\hat{\rho}}_{0}:=|\psi \rangle\langle\psi|$, 
\begin{equation}
	\big(\boldsymbol{\hat{H}}-\alpha\mathbb{I}\big) \big|\psi\rangle= e^{s_{\psi}[\boldsymbol{\hat{\rho}}_{0}, \boldsymbol{\hat{H}}]}\big|\psi\rangle~,
\end{equation}
for 
\begin{equation}
	s_{\psi}:=\frac{-1}{\sqrt{V_{\psi}}} \arccos\Bigg(\frac{E_{\Psi}-\alpha}{\sqrt{V_{\psi}}+(E_{\psi}-\alpha)^2}\Bigg)~,
\end{equation}
where
\begin{equation}
E_{\psi}=\langle \psi|\boldsymbol{\hat{H}}|\psi\rangle, \;\; V_{\psi}= \langle \psi|\boldsymbol{\hat{H}}^{2}|\psi\rangle -E_{\Psi}^{2}~.
\end{equation}
\end{definition2}

The costly part of the QITE algorithm involves the estimation of the non-unitary dynamics generated by a slew of local-Hamiltonians $\boldsymbol{\hat{h}}:=\sum_{m}\alpha_{m}\boldsymbol{\hat{\sigma}}_{m}$, where the $\boldsymbol{\hat{\sigma}}_{m}$ are Pauli strings. The goal of QITE is to find a Hermitian operator $\boldsymbol{\hat{A}}$, dependent on the state $|\psi\rangle$, such that 
\begin{equation}\label{eqn:uni}
	c^{-1/2}e^{-\Delta\tau\boldsymbol{\hat{h}}}|\psi\rangle=e^{-i\Delta \tau\boldsymbol{\hat{A}}}|\psi\rangle~,
\end{equation}
for any $|\psi\rangle\in\mathbb{C}^{2}$ where the l.h.s. is the ITE step Eq.~\eqref{eq:ITE_step} with $c:= \langle \psi|e^{-2\Delta\tau\boldsymbol{\hat{h}}}|\psi\rangle$ and $\Delta\tau$ a small parameter. Smallness of $\Delta\tau$ here means that
\begin{equation}
	c^{-1/2}e^{-\Delta\tau\boldsymbol{\hat{h}}}|\psi\rangle \approx	\|(\mathbb{I}-\Delta\tau\boldsymbol{\hat{h}})|\psi\rangle\|^{-1/2}(\mathbb{I}-\Delta\tau\boldsymbol{\hat{h}})|\psi\rangle~.
\end{equation}
To find a Hermitian matrix $\boldsymbol{\hat{A}}$ satisfying Eq.~\eqref{eqn:uni} we minimize the following norm
\begin{equation}
	\||\Delta_{0}\rangle-|\Delta\rangle\|~,
\end{equation}
where 
\begin{equation}
	|\Delta_{0}\rangle  = \frac{\|\mathbb{I}-\Delta\tau\boldsymbol{\hat{h}}\|^{-1/2}(\mathbb{I}-\Delta\tau\boldsymbol{\hat{h}})|\psi\rangle-|\psi\rangle  }{\Delta\tau}~,
\end{equation}
and 
\begin{equation}
	|\Delta\rangle=  -i\boldsymbol{\hat{A}} |\psi\rangle ~.
\end{equation}
The idea of QITE is that in a small neighborhood of $\Delta\tau=0$ the state $\|(\mathbb{I}-\Delta\tau\boldsymbol{\hat{h}})|\psi\rangle\|^{-1/2}(\mathbb{I}-\Delta\tau\boldsymbol{\hat{h}})|\psi\rangle$ evolves as a one-parameter unitary dynamics over the state $|\psi\rangle$.
Let us now consider
\begin{equation}
	(\boldsymbol{\hat{H}} -\alpha\mathbb{\hat{I}})|\psi\rangle~,
\end{equation}
where 
\begin{align}
	\alpha&=-\|(\mathbb{I}-\Delta\tau\boldsymbol{\hat{h}}) |\psi\rangle\|^{-1/2} ~, \\
	\boldsymbol{\hat{H}}&=-\Delta\tau\|\mathbb{I}-\Delta\tau\boldsymbol{\hat{h}}\|^{-1/2}\boldsymbol{\hat{h}}~.
\end{align}
Noting that $V_{\Psi}=\Delta\tau^{2}\|\mathbb{I}-\Delta\tau\boldsymbol{\hat{h}}\|(\langle\psi|\boldsymbol{\hat{h}}^{2}|\psi\rangle-\langle\psi|\boldsymbol{\hat{h}}|\psi\rangle^{2})$. The latter definitions coupled with Lemma \ref{eqn:suzukilemm} lead to 
\begin{equation}
	\|(\mathbb{I}-\Delta\tau\boldsymbol{\hat{h}})|\psi\rangle\|^{-1/2}(\mathbb{I}-\Delta\tau\boldsymbol{\hat{h}})|\psi\rangle   =
e^{\beta_{\Psi}[\Psi,\boldsymbol{\hat{h}}]}|\psi\rangle~,
\end{equation}
where 
\begin{equation}
	\beta_{\Psi}:=\frac{-\Delta\tau\|\mathbb{I}-\Delta\tau\boldsymbol{\hat{h}}\|^{-1/2}}{\sqrt{V_{\Psi}}}  \nonumber 
		     \arccos\Bigg(\frac{E_{\Psi}-\alpha}{\sqrt{V_{\Psi}}+(E_{\Psi}-\alpha)^2}\Bigg)=
		     \frac{-\arccos\Bigg(\frac{E_{\Psi}-\alpha}{\sqrt{V_{\Psi}}+(E_{\Psi}-\alpha)^2}\Bigg)} {\sqrt{(\langle\psi|\boldsymbol{\hat{h}}^{2}|\psi\rangle-\langle\psi|\boldsymbol{\hat{h}}|\psi\rangle^{2})}}~.
\end{equation}
Whence, the numerical derivative $|\Delta_{0}\rangle$ is approximately the derivative of a one-parameter unitary group, in this case said parameter is a function of $\Delta\tau$ which is non-decreasing with respect to $\Delta\tau$; this has been shown in \cite{glu}. Namely, 
\begin{equation}
|\Delta_{0}\rangle\approx \partial_{\Delta\tau}e^{\beta_{\Psi}(\Delta\tau)[\Psi, \boldsymbol{\hat{h}}]}\big|_{\Delta\tau=0}|\psi\rangle=
t[\Psi, \boldsymbol{\hat{h}}]|\psi\rangle~,
\end{equation}
where $t=\partial_{\Delta\tau}\beta_{\Psi}(\Delta\tau)\big|_{\Delta\tau=0}$; the error here being $\mathcal{O}(\Delta\tau)$.
We therefore conclude that 
\begin{equation}
	t[\boldsymbol{\hat{\rho}}_{0}, \boldsymbol{\hat{h}}]=-i\Delta\tau\boldsymbol{\hat{A}}~,
\end{equation}
i.e., 
\begin{equation}
 	\big\| |\Delta_{0}\rangle-|\Delta\rangle\big\|
=\mathcal{O}(\Delta\tau)  ~,
\end{equation}
whenever $\boldsymbol{\hat{A}}=\frac{it}{\Delta\tau}[\boldsymbol{\hat{\rho}}_{0}, \boldsymbol{\hat{h}}]$. Notice that this points in the same direction as a single iteration of DB-ITE, leading us to conclude that; with $\mathcal{O}(\Delta\tau)$ error, both QITE and DB-ITE produce the one-parameter group generating the geodesic from the initial state $|\psi\rangle$ to the ground state of the Hamiltonian $\boldsymbol{\hat{h}}$ so long as the initial state has overlap with the ground state. 
Of course, a technique may be devised so that the output for $\boldsymbol{\hat{A}}$ is indeed $\frac{it}{\Delta\tau}[\boldsymbol{\hat{\rho}}_{0}, \boldsymbol{\hat{h}}]$ such as the application of Lemma \ref{eqn:suzukilemm}.

With the latter, it can be argued that, in the case of a single qubit, a single iteration of QITE yields the one-parameter group that traces out the geodesic connecting the selected initial state to the ground state. 
We will not do it here but similar techniques may be used to show that multiple QITE iterations in such a case lead to elements of the mentioned one-parameter. Leading to a QITE version of Corollary \ref{eqn:3232} presented below. Before presenting this Corollary, let us first present a lemma proven in \cite{arxivgluzagrover}.

\begin{definition2}[Equivalence of ITE and commutator flow for projector Hamiltonians \cite{arxivgluzagrover}]
\label{eqn:ITEGEO}
Let $\boldsymbol{\hat{P}}$ be an $N$ dimensional. Then, for any ITE evolution time $\tau$ , there exists a time
duration $s_{\tau}$ such that
\begin{equation}\label{eqn:iteinlemm}
	\frac{e^{\tau\boldsymbol{\hat{P}}}|\psi_{0}\rangle}{\|e^{\tau\boldsymbol{\hat{P}}}|\psi_{0}\rangle\|} = e^{s_{\tau}[\boldsymbol{\hat{P}}, \boldsymbol{\hat{\rho}}_{0}]}|\psi_{0}\rangle~,
\end{equation}
where $\frac{ds_{\tau}}{d\tau}\geq 0$, i.e., $s_{\tau}$ is non decreasing, and $\boldsymbol{\hat{\rho}}_{0}:=|\psi_{0}\rangle\langle\psi_{0}|$.
\end{definition2}
This lemma may be extended to the following more generic result via some almost trivial arguments. We present this as a corollary. 
\begin{Co}{Equivalence of ITE and commutator flow for Rank-2 Hamiltonians } 
\label{eqn:3232}
Let $\boldsymbol{\hat{H}}$ be a Rank-2 Hermitian matrix. Then, for any ITE evolution time $\tau$ , there exists a time
duration $s_{\tau}$ such that
\begin{equation}
\frac{e^{\tau\boldsymbol{\hat{H}}}|\psi_{0}\rangle}{\|e^{\tau\boldsymbol{\hat{H}}}|\psi_{0}\rangle\|} = e^{s_{\tau}[\boldsymbol{\hat{H}}, \psi_{0}]}|\psi_{0}\rangle~,
\end{equation}
where $\psi_{0}:=|\psi_{0}\rangle\langle\psi_{0}|$.
\end{Co}

\section{Upper and lower bounds on ITE convergence time}\label{sec:app_Length}

%add index for degeneracy, 'l'
%if it's N-qubit then it's 2^N dimension.

Let $\ket{\psi_0}$ be a $N$-qubit state in a $d=2^N$-dimensional Hilbert space and $\boldsymbol{\hat{H}}$ a Hamiltonian. One can express $\ket{\psi_0}$ in terms of the eigenbasis of  $\boldsymbol{\hat{H}}$ as
\begin{equation}
     \ket{\psi_0} = \sum_{i=0}^{d-1} c_i \ket{E_i}~,
\end{equation}
where $\ket{E_i}$ is the $i$-th element of the eigenbasis, and $c_i = \braket{E_i|\psi_0}$ are their corresponding amplitudes. For simplicity, the groundstate $\ket{E_0}$ is considered to be non-degenerated\footnote{Any possible degeneracies of higher energy levels are not distinguished with a different index.}.
Building on Section \ref{subsec:ITE}, the evolution in imaginary time of $\ket{\psi_0}$ expressed in the eigenbasis of $\boldsymbol{\hat H}$ is expressed as
\begin{equation}
    \ket{\psi(\tau)} = \frac{e^{-\boldsymbol{\hat{H}} \tau} \ket{\psi_0}}{\left\| e^{-\boldsymbol{\hat{H}} \tau} \ket{\psi_0} \right\|} = \frac{\sum_{i=0}^{d-1} c_i e^{-E_i \tau} \ket{E_i}}{\left( \sum_{i=0}^{d-1} |c_i|^2 e^{-2 E_i \tau} \right)^{1/2}}~.
\end{equation}
To bound the time of convergence to the ground state one may compute the fidelity at a given $\tau$ by taking the overlap squared with respect to $\ket{E_0}$
\begin{equation}\label{eq:F_ite}
    F(\tau) = \left| \braket{E_0 | \psi(\tau)} \right|^2 = \frac{|c_0|^2 e^{-2 E_0 \tau}}{\sum_{i=0}^{d-1} |c_i|^2 e^{-2 E_i \tau}}= \frac{|c_0|^2}{|c_0|^2 + \sum\limits_{i=1}^{d-1} |c_i|^2 e^{-2 (E_i - E_0) \tau}}~,
\end{equation}
where the fact that $\braket{E_i|E_j}= \delta_{i,j}$ has been used.
Given this expression, and using the fact that, by construction, $E_0 < E_1 \leq ... \leq E_{d-1}$, then
\begin{equation}
    \frac{|c_0|^2}{|c_0|^2 + \sum\limits_{i=1}^{d-1} |c_i|^2 e^{-2 (E_1 - E_0) \tau}}
    \le F(\tau) =
    \frac{|c_0|^2}{|c_0|^2 + \sum\limits_{i=1}^{d-1} |c_i|^2 e^{-2 (E_i - E_0) \tau}}
    \le
    \frac{|c_0|^2}{|c_0|^2 + \sum\limits_{i=1}^{d-1} |c_i|^2 e^{-2 (E_{d-1} - E_0) \tau}}~,
\end{equation}
which, by imposing normalization to the elements of the sum, yields
\begin{equation}\label{eq:sandwich}
    \frac{|c_0|^2}{|c_0|^2 + (1 - |c_0|^2)\, e^{-2 (E_1 - E_0)\tau}}
    \le
    F(\tau)
    \le
    \frac{|c_0|^2}{|c_0|^2 + (1 - |c_0|^2)\, e^{-2 (E_{d-1} - E_0)\tau}}~.
\end{equation}
Since we aim to fix the final and initial fidelities for our derivation, let us denote $F(\tau) = F_{\mathrm{conv}}$ and $|c_0|^2 = F_0$ for convenience. Solving the bounds in Eq.~(\ref{eq:sandwich}) for $\tau$ it follows that 
\begin{equation}
    \frac{1}{2(E_{d-1} - E_0)} \left[
        \ln\left( \frac{1 -  F_0}{ F_0} \right)
        + \ln\left( \frac{F_{\mathrm{conv}}}{1 - F_{\mathrm{conv}}} \right)
    \right]
    \le
    \tau_{\mathrm{conv}}
    \le
    \frac{1}{2(E_1 - E_0)} \left[
        \ln\left( \frac{1 -  F_0}{ F_0} \right)
        + \ln\left( \frac{F_{\mathrm{conv}}}{1 - F_{\mathrm{conv}}} \right)
    \right].
\end{equation}
%F=1-\epsilon
These bounds depend explicitly on the smallest and largest spectral gaps of the system, $E_1 - E_0$ and $E_{d-1} - E_0$ respectively, as well as on the initial and target fidelities, $ F_0$ and $F_{\mathrm{conv}}$, both defined with respect to the ground state. Notably, the dependence on these fidelities appears in the form of logarithmic ratios: the initial fidelity enters inversely as $\ln[(1 -  F_0)/ F_0]$, diverging as $ F_0 \rightarrow 0$ (i.e. initial states that are orthogonal to the ground state never converge), while the target fidelity contributes via $\ln[F_{\mathrm{conv}} / (1 - F_{\mathrm{conv}})]$, which diverges as $F_{\mathrm{conv}} \rightarrow 1$ (i.e. ITE never reaches maximum fidelity, although it approaches the target exponentially).

An important observation arises in the special case of Hamiltonians with spectral cardinaliy 2, where $E_1 = E_2 = \dots = E_{d-1}$. In this scenario, all excited states are energetically degenerate, and thus the upper and lower spectral gaps coincide: $E_1 - E_0 = E_{d-1} - E_0$. As a result, both the left and right-hand sides of the inequality collapse to the same expression. Invoking the squeeze theorem, this implies that the convergence time is exactly determined
\begin{equation}
    \tau_{\mathrm{conv}}(\boldsymbol{\hat{H}}_\text{Spec-2}) =
    \frac{1}{2(E_1 - E_0)} \left[
        \ln\left( \frac{1 -  F_0}{ F_0} \right)
        + \ln\left( \frac{F_{\mathrm{conv}}}{1 - F_{\mathrm{conv}}} \right)
    \right]~.
\end{equation}
From a practical perspective, in the absence of prior knowledge of the ground state or any symmetry of the system, $F_0$ is the fidelity between two random states within the Hilbert space. Consequently, it acts as a random variable whose probability density function (PDF) is well known (see \cite{Haar_pairwise_fidelity})
\begin{equation}
    P(F) = (2^N-1)(1-F)^{2^N-2}~.
\end{equation}
Such PDF corresponds to a Beta distribution with $\alpha=1, \beta = 2^{N-1}$ \cite{beta_distribution}, and can be leveraged to compute the expectation value of our bounds for the convergence time
\begin{align}
\mathbb{E}( \tau_{\mathrm{conv}})
&\le \int_0^1 \frac{1}{2(E_1 - E_0)}
\left[
\ln \left(  \frac{1 -  F_0}{ F_0} \right)
+ \ln\left( \frac{F_{\mathrm{conv}}}{1 - F_{\mathrm{conv}}} \right)
\right] P(F_0)\, dF_0 \\[4pt]
&= \frac{1}{2(E_1 - E_0)}
\left[
\int_0^1 \ln \left(  \frac{1 -  F_0}{ F_0} \right) P(F_0)\, dF_0
+ \ln\left( \frac{F_{\mathrm{conv}}}{1 - F_{\mathrm{conv}}} \right)
\underbrace{\int_0^1 P(F_0)\, dF_0}_{=\,1}
\right] \\[6pt]
&= \frac{1}{2(E_1 - E_0)}
\Bigg[
\int_0^1 \ln(1-F_0)\,P(F_0)\,dF_0
-\int_0^1 \ln(F_0)\,P(F_0)\,dF_0
+ \ln\left( \frac{F_{\mathrm{conv}}}{1 - F_{\mathrm{conv}}} \right)
\Bigg] \\[8pt]
&= \frac{1}{2(E_1 - E_0)}
\Bigg[
\underbrace{\mathbb{E}_{\mathrm{Beta}(\alpha,\,\beta)}[\ln(1-F_0)]}_{=\ \psi(\beta)-\psi(\alpha+\beta)}
-
\underbrace{\mathbb{E}_{\mathrm{Beta}(\alpha,\,\beta)}[\ln(F_0)]}_{=\ \psi(\alpha)-\psi(\alpha+\beta)}
+ \ln\left( \frac{F_{\mathrm{conv}}}{1 - F_{\mathrm{conv}}} \right)
\Bigg]_{\ \alpha=1,\ \beta=2^N-1} \\[8pt]
&= \frac{1}{2(E_1 - E_0)}
\Bigg[
\big(\psi(\beta)-\psi(\alpha+\beta)\big)
-\big(\psi(\alpha)-\psi(\alpha+\beta)\big)
+ \ln\left( \frac{F_{\mathrm{conv}}}{1 - F_{\mathrm{conv}}} \right)
\Bigg]_{\ \alpha=1,\ \beta=2^N-1} \\[8pt]
&= \frac{1}{2(E_1 - E_0)}
\left[
\psi(2^N-1)-\psi(1)
+ \ln\left( \frac{F_{\mathrm{conv}}}{1 - F_{\mathrm{conv}}} \right)
\right]~,
\end{align}
where we exploited the normalization of $P(F_0)$ and the standard result for the logarithm of the geometric mean of the Beta distribution in terms of the digamma function $\psi(z) \equiv \frac{d}{dz}\ln\Gamma(z)
= \frac{\Gamma'(z)}{\Gamma(z)}$  \cite{beta_distribution}  to solve the different integral terms. 

Now, using the asymptotic expansion of the digamma function for large argument,
\begin{equation}
\psi(z) = \ln z+O\!\left(z^{-1}\right)~,\qquad z\to\infty~,
\end{equation}
we obtain
\begin{align}
\psi(2^N-1)
&= \ln(2^N-1)+O\!\left((2^N-1)^{-1}\right) \nonumber\\
&= N\ln 2+O\!\left(2^{-N}\right)~, \qquad N\to\infty~.
\end{align}
Substituting this into the bound for the expectation value yields
\begin{equation}
\mathbb{E}(\tau_{\mathrm{conv}})
\le
\frac{1}{2(E_1-E_0)}
\left[
N\ln 2 + O\!\left(2^{-N}\right)
-\psi(1)
+\ln\!\left(\frac{F_{\mathrm{conv}}}{1-F_{\mathrm{conv}}}\right)
\right]~,
\qquad N\to\infty~.
\end{equation}
Since $\psi(1)=-\gamma \approx  -0.577$ and $\ln\!\big(\tfrac{F_{\mathrm{conv}}}{1-F_{\mathrm{conv}}}\big)$ are constants\footnote{Note that $\gamma$ is the Euler-Mascheroni constant \cite{beta_distribution}.} with respect to $N$, while the remainder term decays exponentially as $O(2^{-N})$, the leading contribution is linear in $N$. Therefore, we get the asymptotic upper bound
\begin{equation}
    \mathbb{E}(\tau_{\mathrm{conv}})\in O\!\left(\frac{N}{E_1-E_0}\right)~,
\end{equation}
and, analogously from the lower bound
\begin{equation}
    \mathbb{E}(\tau_{\mathrm{conv}}) \in \Omega\left( \frac{N}{E_{d-1} - E_0} \right)~.
\end{equation}
These results establish that the average convergence time required to reach a fixed target fidelity from a random initial state is not only constrained by the spectral properties of the Hamiltonian, but also scales with the dimensionality of the system.

\section{Transpilation and Gate Counts}
\label{sec:app_transpilation}
In this section, we explain how we transpile the QITE and ACQ unitaries and show that, for a single step of the evolution, the gate cost of QITE and ACQ is the same. When combined with the fact that ACQ reduces the total number of evolution steps, this implies that ACQ achieves a substantially reduced overall circuit depth.

For a single time step of QITE, we transpile the unitary in \eqref{eq:QITEunitary}, whereas for a single time step of ACQ we transpile  a unitary of the form \eqref{eq:CompressedU}, with the time parameter $\tau=\Delta \tau_n$ determined by the algorithm. In both cases, we ultimately need to transpile unitaries of the form
\begin{equation}\label{eq:unitaryEvo}
    \boldsymbol{\hat U} = e^{-i t \boldsymbol{\hat A}}~,
\end{equation}
where $\boldsymbol{\hat A}=\sum_{I} a_I \boldsymbol{\hat \sigma}_I$ is a Hermitian operator expressed as a sum of Pauli terms of locality at most $D$, and $t$ is the time step parameter. 

Exact synthesis of a general unitary operator incurs in an exponential scaling in the number of gates \cite{KG,QSD,ZXZ}. 
For instance, for a $D$ local unitary the cost in gates is $O(4^D)$.  This overhead is typically overcome by trading accuracy for a reduced gate count. Specifically, a first-order Trotter decomposition with a single repetition is performed, decomposing \eqref{eq:unitaryEvo} into a product of individual Pauli exponentials of the form
\begin{equation}\label{eq:pauliunitary}
    \boldsymbol{\hat U} \approx \prod_I e^{-it a_I \boldsymbol{\hat \sigma}_I}~.
\end{equation}
Each of these terms can then be transpiled exactly using single-qubit rotations \texttt{u3} and CNOT gates $\texttt{cx}$ \cite{PauliEvo}. The gate count for a single term scales as $O(L)$, where $L\leq D$ is the locality of the Pauli operator $\boldsymbol{\hat \sigma}_I$. Consequently, the total gate count for transpiling \eqref{eq:unitaryEvo} using this method scales as $O(N_PD)$, up to further reductions from optimization routines, as discussed in \cite{pauliunitary}. Here $N_P$ is the number of Pauli terms in the product of \eqref{eq:pauliunitary}.
More precisely, the original QITE work \cite{Motta2019DeterminingEA} shows that the number of Pauli terms involved in a time step of QITE is $N_P\sim4^D/2$. As a result, this decomposition scales as $O(N_K 4^D D/2)$ where $N_K$ is the number of Hamiltonian pieces in the decomposition of \eqref{eq:HamPieces}.

By applying this method to all the $A_{n,k}=\sum_I a_{k,I}(n) \sigma_{n,I}$ of a step $n$, the transpiled ACQ unitary would become
\begin{equation}\label{eq:transpiledVn}
    V_n(t) \approx \prod_{k=1}^{N_K}\prod_{I} e^{-i a_{k,I} (n) \sigma_{k,I} t}~,
\end{equation}
where $V_n(t)$ remains a 1-parameter unitary and the same ACQ algorithm where energy is probed at different times can still be applied. Performing the evolution with this unitary would yield different dynamics that the original one, but as seen in Fig.~\ref{fig:transpilation}, the error introduced by these extra Trotterizations only amounts to a small relative error in the fidelity reached by the quantum circuit.

\begin{figure}
    \centering
    \includegraphics[width=\linewidth]{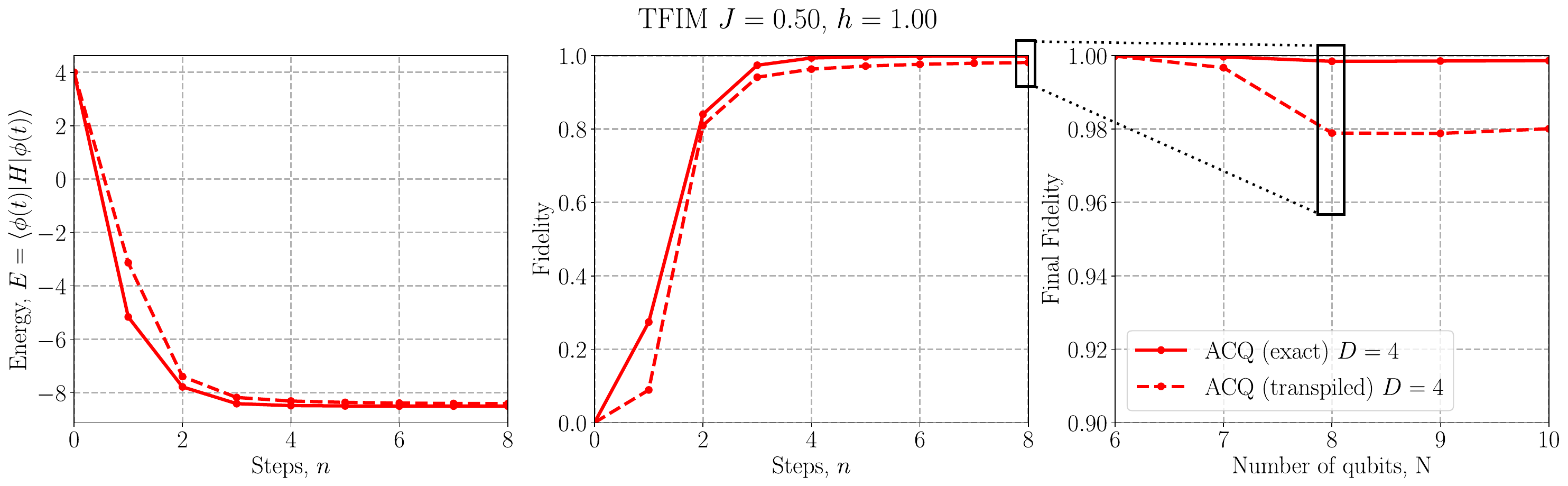}
    \caption{Comparison of the evolution of the energy and fidelity (in the left and central panel respectively) between the exact ACQ evolution, using the unitary in Eq.~\eqref{eq:CompressedU}, and the one generated by the transpiled circuit, using the unitary in Eq.~\eqref{eq:transpiledVn}. We employed the TFIM in the disordered phase to make this comparison for a system of $N=8$ qubits. On the right panel, we show how the final fidelity reached by ACQ and its transpiled version compares for an increasing number of qubits. Notice the change in scale of the fidelity in the right panel.}
    \label{fig:transpilation}
\end{figure}

To achieve what we discussed above and produce the plots in Fig. \ref{fig:comp_DBF} and tables in \ref{tab:resourcecost_ACQDB} we use Qrisp library \cite{seidel2024qrisp}. To build the unitary \eqref{eq:transpiledVn} we use Qrisp's trotterization function on \eqref{eq:unitaryEvo}. For DB-QITE, in both the Heisenberg model and TFIM, we use the exact same trotterization scheme used in \cite{glu}; however for all ACQ transpilations we use a single step of order one trotterization with no time step rescaling.

For the HVA state in the Heisenberg model, we again use the exact same construction as in \cite{glu}. For the HVA state in the TFIM we employ a very similar construction, where we use an ansatz composed of three block-commuting pieces of the Hamiltonian. Specifically, the total TFIM Hamiltonian \eqref{eq:TFIM} is partitioned into even nearest-neighbor $ZZ$ interactions, odd nearest-neighbor $ZZ$ interactions, and the transverse magnetic field $X$ terms:
\begin{equation}
    \boldsymbol{\hat{H}_0} = J \sum_{j \text{ even}} \boldsymbol{\hat{\sigma}}^z_j \boldsymbol{\hat{\sigma}}^z_{j+1}~, \quad
    \boldsymbol{\hat{H}_1} = J \sum_{j \text{ odd}} \boldsymbol{\hat{\sigma}}^z_j \boldsymbol{\hat{\sigma}}^z_{j+1}~, \quad
    \boldsymbol{\hat{H}_2} = h \sum_{j} \boldsymbol{\hat{\sigma}}^x_j~.
\end{equation}

Because the Pauli operators within each individual sub-Hamiltonian mutually commute, their corresponding evolution operators can be compiled into a quantum circuit exactly. The HVA is then constructed by sequentially applying parametrized exponential layers of these sub-Hamiltonians to the initial state $|0\rangle^{\otimes n}$:
\begin{equation}
    |\psi(\boldsymbol{\theta})\rangle = \prod_{d=1}^{p} \left( e^{-i \theta_{d, 2} \boldsymbol{\hat{H}_2}} e^{-i \theta_{d, 1} \boldsymbol{\hat{H}_1}} e^{-i \theta_{d, 0} \boldsymbol{\hat{H}_0}} \right) |0\rangle^{\otimes n},
\end{equation}

The main takeaway is that both QITE and ACQ can be transpiled with a gate-count scaling that is exponential only in the locality parameter $D$ and linear with the number of steps.
When combined with ACQ’s substantial reduction in the number of required steps, this makes ACQ a particularly promising state-preparation method for current and near-term quantum devices.

%
%\section{Geodesics in The Complex Projective Plane}
%\label{eqn:appgeo}
%\input{app_geodesics.tex}
%
%\section{Riemannian Gradient Descent and ITE}
%\label{sec:GradientDescent}
%\input{app_grad_demo.tex}
%
%\section{Proof of Lemma \ref{eqn:2disgeo}}
%\label{sec:2disgeo}
%\input{app_grad_demo.tex}
%

\bibliography{biblio}

@article{Motta2019DeterminingEA,
  title={Determining eigenstates and thermal states on a quantum computer using quantum imaginary time evolution},
  author={Mario Motta and Chong Sun and Adrian T.K. Tan and Matthew J. O'Rourke and Erika Ye and Austin J. Minnich and Fernando G. S. L. Brand{\~a}o and Garnet Kin-Lic Chan},
  journal={Nature Physics},
  year={2019},
  volume={16},
  pages={205 - 210},
  url={https://api.semanticscholar.org/CorpusID:208174822}
}

@inproceedings{pauliunitary,
author = {Li, Gushu and Wu, Anbang and Shi, Yunong and Javadi-Abhari, Ali and Ding, Yufei and Xie, Yuan},
title = {Paulihedral: a generalized block-wise compiler optimization framework for Quantum simulation kernels},
year = {2022},
isbn = {9781450392051},
publisher = {Association for Computing Machinery},
address = {New York, NY, USA},
url = {https://doi.org/10.1145/3503222.3507715},
doi = {10.1145/3503222.3507715},
booktitle = {Proceedings of the 27th ACM International Conference on Architectural Support for Programming Languages and Operating Systems},
pages = {554-569},
numpages = {16},
location = {Lausanne, Switzerland},
series = {ASPLOS '22}
}

@article{mcardle2019variational,
  title={Variational ansatz-based quantum simulation of imaginary time evolution},
  author={McArdle, Sam and Jones, Tyson and Endo, Suguru and Li, Ying and Benjamin, Simon C and Yuan, Xiao},
  journal={npj Quantum Information},
  volume={5},
  number={1},
  pages={75},
  year={2019},
  publisher={Nature Publishing Group UK London}
}

@article{Matt,
  title={On Lieb-Robinson Bounds for the Double Bracket Flow},
  author={Matthew B. Hastings},
  journal={	arXiv:2201.07141},
  volume={},
  number={},
  pages={10},
  year={2022},
  publisher={Arxiv}
}

@article{arxivgluzagrover,
    title={Grover's algorithm is an approximation of imaginary-time evolution}, 
    author={Yudai Suzuki and Marek Gluza and Jeongrak Son and Bi Hong Tiang and Nelly H. Y. Ng and Zoë Holmes},
    year={2025},
    eprint={2507.15065},
    journal={arXiv preprint arXiv:507.15065},
    primaryClass={quant-ph},
    url={https://arxiv.org/abs/2507.15065}, 
}

@article{Hadfield2021,
author = {Hadfield, Stuart},
title = {On the Representation of Boolean and Real Functions as Hamiltonians for Quantum Computing},
year = {2021},
issue_date = {December 2021},
publisher = {Association for Computing Machinery},
address = {New York, NY, USA},
volume = {2},
number = {4},
url = {https://doi.org/10.1145/3478519},
doi = {10.1145/3478519},
journal = {ACM Transactions on Quantum Computing},
month = dec,
articleno = {18},
numpages = {21}
}

@ARTICLE{Lucas2014,
AUTHOR={Lucas, Andrew },
TITLE={Ising formulations of many NP problems},
JOURNAL={Frontiers in Physics},
VOLUME={2},
YEAR={2014},
URL={https://www.frontiersin.org/journals/physics/articles/10.3389/fphy.2014.00005},
DOI={10.3389/fphy.2014.00005},
ISSN={2296-424X},
}

@incollection{NAKATANI2018,
title = {Matrix Product States and Density Matrix Renormalization Group Algorithm},
booktitle = {Reference Module in Chemistry, Molecular Sciences and Chemical Engineering},
publisher = {Elsevier},
year = {2018},
isbn = {978-0-12-409547-2},
doi = {https://doi.org/10.1016/B978-0-12-409547-2.11473-8},
url = {https://www.sciencedirect.com/science/article/pii/B9780124095472114738},
author = {Naoki Nakatani},
}

@article{farhi2014quantum,
  title={A quantum approximate optimization algorithm},
  author={Farhi, Edward and Goldstone, Jeffrey and Gutmann, Sam},
  journal={arXiv preprint arXiv:1411.4028},
  year={2014}
}

@article{RevModPhys.93.045003,
  title = {Matrix product states and projected entangled pair states: Concepts, symmetries, theorems},
  author = {Cirac, J. Ignacio and P\'erez-Garc\'{\i}a, David and Schuch, Norbert and Verstraete, Frank},
  journal = {Rev. Mod. Phys.},
  volume = {93},
  issue = {4},
  pages = {045003},
  numpages = {65},
  year = {2021},
  month = {Dec},
  publisher = {American Physical Society},
  doi = {10.1103/RevModPhys.93.045003},
  url = {https://link.aps.org/doi/10.1103/RevModPhys.93.045003}
}

@misc{farhi2000,
      title={Quantum Computation by Adiabatic Evolution}, 
      author={Edward Farhi and Jeffrey Goldstone and Sam Gutmann and Michael Sipser},
      year={2000},
      eprint={quant-ph/0001106},
      archivePrefix={arXiv},
      primaryClass={quant-ph},
      url={https://arxiv.org/abs/quant-ph/0001106}, 
}

@online{n,
  author = {Notes},
  title = {in Maris' personal website},
  url = {http://home.lu.lv/~sd20008/papers/Bloch},
}

@article{mott,
  author = {Mario Motta and Chong Sun and Adrian T. K. Tan and Matthew J. O'Rourke and Erika Ye and Austin J. Minnich and Fernando G. S. L. Brandão and Garnet Kin-Lic Chan},
  year = {2020},
  title = {"Determining eigenstates and thermal states on a quantum computer using quantum imaginary time evolution"},
  journal = {Nature Physics},
  volume = {16},
  pages = {s205-210},
  note = {page},
}

@article{glu,
 title = {Double-Bracket Quantum Algorithms for Quantum Imaginary-Time Evolution},
  author = {Gluza, Marek and Son, Jeongrak and Tiang, Bi Hong and Zander, Ren\'e and Seidel, Raphael and Suzuki, Yudai and Holmes, Zo\"e and Ng, Nelly H. Y.},
  journal = {Phys. Rev. Lett.},
  volume = {136},
  issue = {2},
  pages = {020601},
  numpages = {12},
  year = {2026},
  month = {Jan},
  publisher = {American Physical Society},
  doi = {10.1103/rw81-k8vk},
  url = {https://link.aps.org/doi/10.1103/rw81-k8vk}
}

@misc{na,
      title={Equating quantum imaginary time evolution, Riemannian gradient flows, and stochastic implementations}, 
      author={Nathan A. McMahon and Mahum Pervez and Christian Arenz},
      year={2025},
      eprint={2504.06123},
      archivePrefix={arXiv},
      primaryClass={quant-ph},
      url={https://arxiv.org/abs/2504.06123}, 
}

@article{arxisuzukisignaling, 
  title = {Double-bracket algorithm for quantum signal processing without post-selection},
  author = {Yudai Suzuki Bi Hong Tiang, Jeongrak Son Nelly H. Y. Ng Zoë Holmes Marek Gluza},
  journal = {arxiv preprint},
  volume = {},
  issue = {},
  pages = {},
  numpages = {45},
  year = {2025},
  month = {April},
  publisher = {Arxiv},
  doi = {
https://doi.org/10.48550/arXiv.2504.01077},
  url = {https://arxiv.org/abs/2504.01077}
}

@article{PhysRevA.111.012424,
  title = {Accelerating quantum imaginary-time evolution with random measurements},
  author = {Kolotouros, Ioannis and Joseph, David and Narayanan, Anand Kumar},
  journal = {Phys. Rev. A},
  volume = {111},
  issue = {1},
  pages = {012424},
  numpages = {17},
  year = {2025},
  month = {Jan},
  publisher = {American Physical Society},
  doi = {10.1103/PhysRevA.111.012424},
  url = {https://link.aps.org/doi/10.1103/PhysRevA.111.012424}
}

@article{PhysRevA.109.052430,
  title = {Combinatorial optimization with quantum imaginary time evolution},
  author = {Bauer, Nora M. and Alam, Rizwanul and Siopsis, George and Ostrowski, James},
  journal = {Phys. Rev. A},
  volume = {109},
  issue = {5},
  pages = {052430},
  numpages = {8},
  year = {2024},
  month = {May},
  publisher = {American Physical Society},
  doi = {10.1103/PhysRevA.109.052430},
  url = {https://link.aps.org/doi/10.1103/PhysRevA.109.052430}
}

@article{MolecularImprov,
author = {Tsuchimochi, Takashi and Ryo, Yoohee and Ten-no, Seiichiro L. and Sasasako, Kazuki},
title = {Improved Algorithms of Quantum Imaginary Time Evolution for Ground and Excited States of Molecular Systems},
journal = {Journal of Chemical Theory and Computation},
volume = {19},
number = {2},
pages = {503-513},
year = {2023},
doi = {10.1021/acs.jctc.2c00906},
URL = { 
        https://doi.org/10.1021/acs.jctc.2c00906
},
}

@misc{xie2025adaptiveweightedqitevqealgorithm,
      title={An Adaptive Weighted QITE-VQE Algorithm for Combinatorial Optimization Problems}, 
      author={Ningyi Xie and Xinwei Lee and Tiejin Chen and Yoshiyuki Saito and Nobuyoshi Asai and Dongsheng Cai},
      year={2025},
      eprint={2504.10651},
      archivePrefix={arXiv},
      primaryClass={quant-ph},
      url={https://arxiv.org/abs/2504.10651}, 
}

@article{Nishi2021,
  title = {Implementation of quantum imaginary-time evolution method on NISQ devices by introducing nonlocal approximation},
  volume = {7},
  ISSN = {2056-6387},
  url = {http://dx.doi.org/10.1038/s41534-021-00409-y},
  DOI = {10.1038/s41534-021-00409-y},
  number = {1},
  journal = {npj Quantum Information},
  publisher = {Springer Science and Business Media LLC},
  author = {Nishi,  Hirofumi and Kosugi,  Taichi and Matsushita,  Yu-ichiro},
  year = {2021},
  month = jun 
}

@article{Gomes2020,
  title = {Efficient Step-Merged Quantum Imaginary Time Evolution Algorithm for Quantum Chemistry},
  volume = {16},
  ISSN = {1549-9626},
  url = {http://dx.doi.org/10.1021/acs.jctc.0c00666},
  DOI = {10.1021/acs.jctc.0c00666},
  number = {10},
  journal = {Journal of Chemical Theory and Computation},
  publisher = {American Chemical Society (ACS)},
  author = {Gomes,  Niladri and Zhang,  Feng and Berthusen,  Noah F. and Wang,  Cai-Zhuang and Ho,  Kai-Ming and Orth,  Peter P. and Yao,  Yongxin},
  year = {2020},
  month = sep,
  pages = {6256–6266}
}

@article{Alam2023,
  title = {Solving MaxCut with quantum imaginary time evolution},
  volume = {22},
  ISSN = {1573-1332},
  url = {http://dx.doi.org/10.1007/s11128-023-04045-7},
  DOI = {10.1007/s11128-023-04045-7},
  number = {7},
  journal = {Quantum Information Processing},
  publisher = {Springer Science and Business Media LLC},
  author = {Alam,  Rizwanul and Siopsis,  George and Herrman,  Rebekah and Ostrowski,  James and Lotshaw,  Phillip C. and Humble,  Travis S.},
  year = {2023},
  month = jul 
}

@article{Stokes2020quantumnatural,
  doi = {10.22331/q-2020-05-25-269},
  url = {https://doi.org/10.22331/q-2020-05-25-269},
  title = {Quantum {N}atural {G}radient},
  author = {Stokes, James and Izaac, Josh and Killoran, Nathan and Carleo, Giuseppe},
  journal = {{Quantum}},
  issn = {2521-327X},
  publisher = {{Verein zur F{\"{o}}rderung des Open Access Publizierens in den Quantenwissenschaften}},
  volume = {4},
  pages = {269},
  month = may,
  year = {2020}
}

@article{wierichs2020avoiding,
  title={Avoiding local minima in variational quantum eigensolvers with the natural gradient optimizer},
  author={Wierichs, David and Gogolin, Christian and Kastoryano, Michael},
  journal={Physical Review Research},
  volume={2},
  number={4},
  pages={043246},
  year={2020},
  publisher={APS}
}

@article{koczor2022quantum,
  title={Quantum natural gradient generalized to noisy and nonunitary circuits},
  author={Koczor, B{\'a}lint and Benjamin, Simon C},
  journal={Physical Review A},
  volume={106},
  number={6},
  pages={062416},
  year={2022},
  publisher={APS}
}

@article{62wx-tvk5,
  title = {Quantum natural gradient optimizer on noisy platforms: Quantum approximate optimization algorithm as a case study},
  author = {Dell'Anna, Federico and G\'omez-Lurbe, Rafael and P\'erez, Armando and Ercolessi, Elisa},
  journal = {Phys. Rev. A},
  volume = {112},
  issue = {2},
  pages = {022612},
  numpages = {16},
  year = {2025},
  month = {Aug},
  publisher = {American Physical Society},
  doi = {10.1103/62wx-tvk5},
  url = {https://link.aps.org/doi/10.1103/62wx-tvk5}
}

@misc{gomezlurbe2025efficientprotocolestimatequantum,
  title        = {Efficient protocol to estimate the Quantum Fisher Information Matrix for Commuting-Block Circuits},
  author       = {Rafael Gomez-Lurbe},
  year         = {2025},
  eprint       = {2505.09818},
  archivePrefix= {arXiv},
  primaryClass = {quant-ph},
  note         = {arXiv:2505.09818},
  url          = {https://arxiv.org/abs/2505.09818}
}

@Article{SciPostPhys.9.4.048,
	title={{Geometry of variational methods: dynamics of closed quantum systems}},
	author={Lucas Hackl and Tommaso Guaita and Tao Shi and Jutho Haegeman and Eugene Demler and J. Ignacio Cirac},
	journal={SciPost Phys.},
	volume={9},
	pages={048},
	year={2020},
	publisher={SciPost},
	doi={10.21468/SciPostPhys.9.4.048},
	url={https://scipost.org/10.21468/SciPostPhys.9.4.048},
}

@article{Iqbal2024topological,
  title        = {Topological order from measurements and feed-forward on a trapped-ion quantum computer},
  author       = {Iqbal, Mohsin and Tantivasadakarn, Nathanan and Gatterman, Thomas M. and Gerber, Justin A. and Gilmore, Kevin and Gresh, Dan and Hankin, Aaron and Hewitt, Nathan and Horst, Chandler V. and Matheny, Mitchell and Mengle, Tanner and Neyenhuis, Brian and Vishwanath, Ashvin and Foss-Feig, Michael and Verresen, Ruben and Dreyer, Henrik},
  journal      = {Communications Physics},
  year         = {2024},
  volume       = {7},
  number       = {1},
  pages        = {205},
  doi          = {10.1038/s42005-024-01698-3},
  url          = {https://www.nature.com/articles/s42005-024-01698-3}
}

@article{understandingQITE,
  title={Understanding Quantum Imaginary Time Evolution and its Variational form},
  author={Andreu, Angl{\'e}s-Castillo and Luca, Ion and Tanmoy, Pandit and Gomez-Lurbe, Rafael and Mart{\'\i}nez, Rodrigo and Garcia-March, Miguel Angel},
  journal={arXiv preprint arXiv:2510.02015},
  year={2025},
  url={https://arxiv.org/abs/2510.02015}
}

@article{boost_MPS,
  title = {Boosted imaginary time evolution of matrix product states},
  author = {Symons, Benjamin C. B. and Manawadu, Dilhan and Galvin, David and Mensa, Stefano},
  journal = {Phys. Rev. B},
  volume = {110},
  issue = {17},
  pages = {174302},
  numpages = {7},
  year = {2024},
  month = {Nov},
  publisher = {American Physical Society},
  doi = {10.1103/PhysRevB.110.174302},
  url = {https://link.aps.org/doi/10.1103/PhysRevB.110.174302}
}

@misc{PauliEvo,
      title={Paulihedral: A Generalized Block-Wise Compiler Optimization Framework For Quantum Simulation Kernels}, 
      author={Gushu Li and Anbang Wu and Yunong Shi and Ali Javadi-Abhari and Yufei Ding and Yuan Xie},
      year={2021},
      eprint={2109.03371},
      archivePrefix={arXiv},
      primaryClass={quant-ph},
      url={https://arxiv.org/abs/2109.03371}, 
}

@article{QSD,
   title={Synthesis of quantum-logic circuits},
   volume={25},
   ISSN={1937-4151},
   url={http://dx.doi.org/10.1109/TCAD.2005.855930},
   DOI={10.1109/tcad.2005.855930},
   number={6},
   journal={IEEE Transactions on Computer-Aided Design of Integrated Circuits and Systems},
   publisher={Institute of Electrical and Electronics Engineers (IEEE)},
   author={Shende, V.V. and Bullock, S.S. and Markov, I.L.},
   year={2006},
   month=jun, pages={1000–1010} }

@misc{ZXZ,
      title={Beyond Quantum Shannon: Circuit Construction for General n-Qubit Gates Based on Block ZXZ-Decomposition}, 
      author={Anna M. Krol and Zaid Al-Ars},
      year={2024},
      eprint={2403.13692},
      archivePrefix={arXiv},
      primaryClass={quant-ph},
      url={https://arxiv.org/abs/2403.13692}, 
}

@misc{KG,
      title={Cartan Decomposition of SU($2^n$), Constructive Controllability of Spin systems and Universal Quantum Computing}, 
      author={Navin Khaneja and Steffen Glaser},
      year={2000},
      eprint={quant-ph/0010100},
      archivePrefix={arXiv},
      primaryClass={quant-ph},
      url={https://arxiv.org/abs/quant-ph/0010100}, 
}

@article{McClean_2018,
   title={Barren plateaus in quantum neural network training landscapes},
   volume={9},
   ISSN={2041-1723},
   url={http://dx.doi.org/10.1038/s41467-018-07090-4},
   DOI={10.1038/s41467-018-07090-4},
   number={1},
   journal={Nature Communications},
   publisher={Springer Science and Business Media LLC},
   author={McClean, Jarrod R. and Boixo, Sergio and Smelyanskiy, Vadim N. and Babbush, Ryan and Neven, Hartmut},
   year={2018},
   month=nov }

@article{Larocca_2025,
   title={Barren plateaus in variational quantum computing},
   volume={7},
   ISSN={2522-5820},
   url={http://dx.doi.org/10.1038/s42254-025-00813-9},
   DOI={10.1038/s42254-025-00813-9},
   number={4},
   journal={Nature Reviews Physics},
   publisher={Springer Science and Business Media LLC},
   author={Larocca, Martín and Thanasilp, Supanut and Wang, Samson and Sharma, Kunal and Biamonte, Jacob and Coles, Patrick J. and Cincio, Lukasz and McClean, Jarrod R. and Holmes, Zoë and Cerezo, M.},
   year={2025},
   month=mar, pages={174–189} }

@article{cirac2012goals,
  title={Goals and opportunities in quantum simulation},
  author={Cirac, J Ignacio and Zoller, Peter},
  journal={Nature Physics},
  volume={8},
  number={4},
  pages={264--266},
  year={2012},
  publisher={Nature Publishing Group UK London},
  doi={10.1038/nphys2275}
}

@article{bernien2017probing,
  title={Probing many-body dynamics on a 51-atom quantum simulator},
  author={Berni{\'e}n, Hannes and Schwartz, Sylvain and Keesling, Alexander and Levine, Harry and Omran, Ahmed and Pichler, Hannes and Choi, Soonwon and Zibrov, Alexander S and Endres, Manuel and Greiner, Markus and others},
  journal={Nature},
  volume={551},
  number={7682},
  pages={579--584},
  year={2017},
  publisher={Nature Publishing Group UK London},
  doi={10.1038/nature24622}
}

@article{mcardle2020quantum,
  title={Quantum computational chemistry},
  author={McArdle, Sam and Endo, Suguru and Aspuru-Guzik, Al{\'a}n and Benjamin, Simon C and Yuan, Xiao},
  journal={Reviews of Modern Physics},
  volume={92},
  number={1},
  pages={015003},
  year={2020},
  publisher={APS},
  doi={10.1103/RevModPhys.92.015003}
}

@article{suzuki1976generalized,
  title={Generalized Trotter's formula and systematic approximants of exponential operators and inner derivations with applications to many-body problems},
  author={Suzuki, Masuo},
  journal={Communications in Mathematical Physics},
  volume={51},
  number={2},
  pages={183--190},
  year={1976},
  publisher={Springer},
  doi={10.1007/BF01609348}
}

@inproceedings{farhi2014quantumMaxCut,
  title={Quantum Supremacy and the Complexity of Random Circuit Sampling},
  author={Farhi, Edward and Goldstone, Jeffrey and Gutmann, Sam},
  booktitle={Proceedings of the 46th Annual ACM Symposium on Theory of Computing},
  pages={76--89},
  year={2014},
  organization={ACM},
  doi={10.1145/2591796.2591853}
}

@article{peruzzo2014variational,
  title={A variational eigenvalue solver on a photonic quantum processor},
  author={Peruzzo, Alberto and McClean, Jarrod and Shadbolt, Peter and Yung, Man-Hong and Zhou, Xiao-Qi and Love, Peter J and Aspuru-Guzik, Al{\'a}n and O’brien, Jeremy L},
  journal={Nature Communications},
  volume={5},
  number={1},
  pages={4213},
  year={2014},
  publisher={Nature Publishing Group UK London},
  doi={10.1038/ncomms5213}
}

@article{preskill2018quantum,
  title={Quantum Computing in the NISQ era and beyond},
  author={Preskill, John},
  journal={Quantum},
  volume={2},
  pages={79},
  year={2018},
  publisher={Verein zur F{\"o}rderung des Open Access Publizierens in den Quantenwissenschaften},
  doi={10.22331/q-2018-08-06-79}
}

@article{wick1954properties,
  title={Properties of Bethe-Salpeter wave functions},
  author={Wick, Gian Carlo},
  journal={Physical Review},
  volume={96},
  number={4},
  pages={1124--1134},
  year={1954},
  publisher={APS},
  doi={10.1103/PhysRev.96.1124}
}

@article{Haar_pairwise_fidelity,
  title = {Average fidelity between random quantum states},
  author = {\ifmmode \dot{Z}\else \.{Z}\fi{}yczkowski, Karol and Sommers, Hans-J\"urgen},
  journal = {Phys. Rev. A},
  volume = {71},
  issue = {3},
  pages = {032313},
  numpages = {11},
  year = {2005},
  month = {Mar},
  publisher = {American Physical Society},
  doi = {10.1103/PhysRevA.71.032313},
  url = {https://link.aps.org/doi/10.1103/PhysRevA.71.032313}
}

@book{beta_distribution,
  title={Continuous Univariate Distributions, Volume 2},
  author={Johnson, N.L. and Kotz, S. and Balakrishnan, N.},
  isbn={9780471584940},
  lccn={93045348},
  series={Wiley Series in Probability and Statistics},
  url={https://books.google.es/books?id=BTANEAAAQBAJ},
  year={1995},
  publisher={Wiley}
}

@article{PRXQuantum.1.020319,
  title = {Exploring Entanglement and Optimization within the Hamiltonian Variational Ansatz},
  author = {Wiersema, Roeland and Zhou, Cunlu and de Sereville, Yvette and Carrasquilla, Juan Felipe and Kim, Yong Baek and Yuen, Henry},
  journal = {PRX Quantum},
  volume = {1},
  issue = {2},
  pages = {020319},
  numpages = {14},
  year = {2020},
  month = {Dec},
  publisher = {American Physical Society},
  doi = {10.1103/PRXQuantum.1.020319},
  url = {https://link.aps.org/doi/10.1103/PRXQuantum.1.020319}
}

@misc{seidel2024qrisp,
  title         = {Qrisp: A Framework for Compilable High-Level Programming of Gate-Based Quantum Computers},
  author        = {Raphael Seidel and Sebastian Bock and Ren\'{e} Zander and Matic Petri{\v c} and Niklas Steinmann and Nikolay Tcholtchev and Manfred Hauswirth},
  year          = {2024},
  eprint        = {2406.14792},
  archivePrefix = {arXiv},
  primaryClass  = {quant-ph},
  url           = {https://arxiv.org/abs/2406.14792}
}

@Article{sym14050996,
AUTHOR = {Sacco Shaikh, Daniel and Sassetti, Maura and Traverso Ziani, Niccolò},
TITLE = {Parity-Dependent Quantum Phase Transition in the Quantum Ising Chain in a Transverse Field},
JOURNAL = {Symmetry},
VOLUME = {14},
YEAR = {2022},
NUMBER = {5},
ARTICLE-NUMBER = {996},
URL = {https://www.mdpi.com/2073-8994/14/5/996},
ISSN = {2073-8994},
DOI = {10.3390/sym14050996}
}
\bibliographystyle{unsrt}

\end{document}